\documentclass[11pt,reqno]{article}

\usepackage[utf8]{inputenc}
\usepackage[T1]{fontenc}
\usepackage{amsfonts,amssymb,amsthm,amsmath,mathrsfs,mathtools}
\usepackage{comment}
\usepackage[authoryear]{natbib}
\usepackage{enumitem,setspace,multicol,multirow}
\usepackage{graphicx,float}
\usepackage{adjustbox}
\usepackage{url}
\usepackage[dvipsnames]{xcolor}
\usepackage[colorlinks=true,linkcolor=cobalt,filecolor=magenta,urlcolor=cobalt,citecolor=cobalt]{hyperref}
\usepackage[capitalise,noabbrev]{cleveref}
\usepackage{booktabs}
\usepackage{cancel}
\usepackage{pdfpages}
\usepackage{rotating}
\usepackage{subcaption}
\usepackage{algorithm}
\usepackage{algpseudocode}
\usepackage[margin=1in]{geometry}
\usepackage{soul}

\graphicspath{{figs/}}

\definecolor{cobalt}{rgb}{0.0,0.25,0.57}
\definecolor{crimson}{rgb}{0.86,0.08,0.24}
\definecolor{darkgreen}{rgb}{0.0,0.4,0.0}
\definecolor{grayred}{rgb}{0.7,0.2,0.2}
\definecolor{bettercolor}{rgb}{0.0,0.25,0.65}

\newtheorem{theorem}{Theorem}
\newtheorem{lemma}[theorem]{Lemma}
\newtheorem{assumption}{Assumption}
\crefname{assumption}{Assumption}{Assumptions}
\newtheorem{proposition}[theorem]{Proposition}

\newtheorem{corollary}[theorem]{Corollary}

\renewcommand{\cal}[1]{\mathcal{#1}}

\newcommand{\del}{\partial}
\newcommand{\eps}{\varepsilon}

\def\R{\mathbb{R}}

\newcommand{\diam}{\mathop{\rm diam}}

\newcommand\numberthis{\addtocounter{equation}{1}\tag{\theequation}}
\renewcommand{\qed}{\hfill$\square$}
\def\E{\mathbb{E}}
\def\P{\mathbb{P}}

\newcommand{\var}{{\rm Var}}

\newcommand{\normal}{\mathcal{N}}

\newcommand{\dto}{\stackrel{d}{\longrightarrow}}
\newcommand{\Pto}{\rightarrow_p}

\def\o{o_p}
\def\O{\mathcal{O}_p}
\providecommand{\argmin}{\mathop{\rm argmin}}

\newcommand{\wh}[1]{\widehat{#1}}
\newcommand{\wt}[1]{\widetilde{#1}}

\newcommand{\mc}[1]{\mathcal{#1}}

\newcommand{\1}{\text{$\mathbf{1}$}}
\newcommand{\ind}[1]{\1{\left\{#1\right\}}}
\newcommand{\inds}[1]{\1_{#1}}
\providecommand{\mmz}{\mathop{\rm{minimize}}}
\providecommand{\mxz}{\mathop{\rm{maximize}}}
\def\t{\tau_c}
\def\g{\gamma}
\def\mtd{1}
\newcommand{\rhos}{\mc{R}}

\newcommand{\firstpagefootnote}[1]{\footnote{#1}}

\newcommand{\ag}[1]{#1}

\newcommand{\good}[1]{\textcolor{darkgreen}{\textbf{#1}}}
\newcommand{\better}[1]{\textcolor{bettercolor}{\textbf{#1}}}
\newcommand{\bad}[1]{\textcolor{grayred}{\textbf{#1}}}

\setstcolor{red}

\title{PLRD: Partially Linear Regression Discontinuity Inference}
\author{Aditya Ghosh\\
  \texttt{ghoshadi@stanford.edu}
  \and
  Guido Imbens\\
  \texttt{imbens@stanford.edu}
  \and
  Stefan Wager\\
  \texttt{swager@stanford.edu}}
\date{Stanford University\\[1em] \today}

\begin{document}

\allowdisplaybreaks
\sloppy
\setlength{\tabcolsep}{5pt}

\maketitle

\begin{abstract}
Regression discontinuity designs have become one of the most popular research designs in empirical economics. We argue, however, that the widely used approaches to building confidence intervals in regression discontinuity designs often exhibit suboptimal behavior in practice. We propose a new estimator, the partially linear regression discontinuity (PLRD) estimator that, in set of a simulation studies carefully calibrated to twelve high-profile applications of regression discontinuity designs, has substantially lower estimation error than available comparison methods. Throughout our experiments, the confidence intervals built using PLRD are both valid and short. We also provide large-sample guarantees for PLRD. Our simulation study serves as a general template for how new econometric methods can be credibly evaluated relative to the existing alternatives by constructing simulation designs that generate synthetic data indistinguishable from the original data using the Wasserstein generative adversarial network methodology.
\end{abstract}

\section{Introduction}

The fraction of empirical NBER working papers using regression discontinuity designs (RDDs) increased from nearly zero in
1995 to between  15\% and 20\% by 2015 \citep[Figure 4B]{currie2020technology}.\firstpagefootnote{Up-to-date numbers are maintained at \url{https://paulgp.com/econlit-pipeline} \citep{goldsmith2024tracking}.}
Despite the ubiquity of RDDs in applied work, however, there are still open questions as to econometric best practices for inference in RDDs.
In particular, we find that default implementations of widely used local-polynomial methods can fail to achieve nominal coverage in realistic settings.

We assess the performance of three commonly used  methods by evaluating them on simulations calibrated to twelve high-profile applications of RDDs.
First, we consider conventional local linear regression \citep{Hahn-et-al-2001} with
Imbens--Kalyanamaran bandwidth \citep{imbens2012optimal}. 
We also consider two more recent proposals \citep{CCT2014robustCI,ArmstrongKolesar2018}, as implemented in the \texttt{R} functions \texttt{rdrobust} \citep{CCT2014robustCI} and \texttt{RDHonest} \citep{ArmstrongKolesar2018}. These methods are representative of the current state of the art for inference in RDDs \citep{stommes2023reliability}, and both offer automated end-to-end confidence interval constructions that can directly be applied to raw data. The method \texttt{rdrobust} \citep{CCT2014robustCI} is a dominant approach to regression-discontinuity inference in current econometric practice.\firstpagefootnote{This
is the only general purpose regression discontinuity package listed in the ``Econometrics
Task View'' for highly used \texttt{R} packages on CRAN (\url{https://cran.r-project.org/web/views/Econometrics.html}); and
both the \texttt{R} \citep{calonico2015rdrobust} and Stata \citep{calonico2017rdrobust} versions of the function
are widely used and the original paper \citep{CCT2014robustCI} cited is over 5000 times as of now.}
The method starts by running a local linear regression \citep{Hahn-et-al-2001}, and then
applies a 2nd-order bias correction to mitigate curvature bias.
Meanwhile, the method \texttt{RDHonest} \citep{ArmstrongKolesar2018} explicitly inflates the
width of local linear regression confidence intervals to accommodate worst-case plausible curvature
bias (rather than trying to correct for this bias).

We use generative adversarial
neural networks \citep{goodfellow2014generative} to generate simulation instances that are
reflective of real-world datasets \citep{athey2021using}.
Specifically, we calibrate our evaluation to eight seminal regression discontinuity studies
\citep{lee2008,cattaneo2015senate,LudwigMiller,meyersson2014islamic,matsudaira,jacob2004,Oreopoulos,Akhtari2022}.\footnote{Some of the studies have more than one plausible outcome variables, leading to a  total of twelve applications.}
The generative adversarial learning approach yields simulators whose
data-generating distribution closely mimics the moments and idiosyncrasies of the original datasets;
see \cref{fig:gan-lee} for an illustration. The generated simulators do not use
explicit parametric models and do not enforce any built-in smoothness assumptions, thus bringing
us towards a neutral comparison ground for evaluating methods that were motivated under different assumptions \citep{athey2021using}.

\begin{table}[t]
\small
\centering
\setlength{\tabcolsep}{2.7pt}
\renewcommand{\arraystretch}{1.05}
\begin{adjustbox}{max width=\textwidth}
\begin{tabular}{@{}l*{12}{c}@{}}
\toprule
\multirow{2}{*}{}  
& \multicolumn{3}{c}{conventional}  
& \multicolumn{3}{c}{\texttt{rdrobust}} 
& \multicolumn{3}{c}{\texttt{RDHonest}} 
& \multicolumn{3}{c}{\texttt{plrd}}\\ 
\cmidrule(lr){2-4} \cmidrule(lr){5-7} \cmidrule(lr){8-10} \cmidrule(lr){11-13} 
\textbf{Dataset} 
& \textbf{cvrg} & \textbf{width} & \textbf{rmse}
& \textbf{cvrg} & \textbf{width} & \textbf{rmse}
& \textbf{cvrg} & \textbf{width} & \textbf{rmse}
& \textbf{cvrg} & \textbf{width} & \textbf{rmse} \\ 
\midrule
\citet{lee2008} & \bad{94.6\%} & 0.036 & 0.010 & \good{94.9\%} & 0.043 & 0.011 & \good{97.3\%} & 0.055 & 0.012 & \good{98.4\%} & 0.044 & \better{0.009} \\
\citet{cattaneo2015senate} & \bad{94.3\%} & 4.955 & 1.316 & \good{94.9\%} & 5.828 & 1.516 & \good{98.0\%} & 6.767 & 1.464 & \good{96.2\%} & 3.865 & \better{0.937} \\
\citet{LudwigMiller} & \bad{94.2\%} & 5.275 & 1.599 & \bad{94.3\%} & 6.247 & 1.883 & \good{97.5\%} & 7.017 & 1.948 & \good{96.5\%} & 3.746 & \better{1.144} \\
\citet{meyersson2014islamic} & \bad{93.9\%} & 6.309 & 1.732 & \bad{94.3\%} & 7.436 & 1.998 & \good{96.8\%} & 7.941 & 1.858 & \good{96.7\%} & 4.682 & \better{1.159} \\
\citet{matsudaira} (reading) & \bad{82.3\%} & 0.050 & 0.018 & \bad{90.1\%} & 0.059 & 0.018 & \good{97.9\%} & 0.104 & 0.022 & \good{96.7\%} & 0.059 & \better{0.014} \\
\citet{matsudaira} (math) & \bad{91.7\%} & 0.044 & 0.013 & \bad{93.1\%} & 0.053 & 0.015 & \good{98.5\%} & 0.093 & 0.019 & \good{97.6\%} & 0.051 & \better{0.011} \\
\citet{jacob2004} (reading) & \bad{93.7\%} & 0.069 & 0.020 & \bad{94.0\%} & 0.083 & 0.023 & \good{97.7\%} & 0.116 & 0.026 & \good{97.2\%} & 0.078 & \better{0.019} \\
\citet{jacob2004} (math) & \bad{94.2\%} & 0.058 & 0.016 & \good{95.2\%} & 0.070 & 0.018 & \good{97.9\%} & 0.122 & 0.026 & \good{97.6\%} & 0.060 & \better{0.013} \\
\citet{Oreopoulos} & \bad{93.9\%} & 0.191 & 0.055 & \good{95.0\%} & 0.258 & 0.072 & \good{99.0\%} & 0.304 & 0.060 & \good{95.9\%} & 0.124 & \better{0.031} \\
\citet{Akhtari2022} (4th grade) & \bad{94.3\%} & 0.104 & 0.028 & \good{95.0\%} & 0.119 & 0.031 & \good{97.3\%} & 0.135 & 0.031 & \good{96.8\%} & 0.096 & \better{0.022} \\
\citet{Akhtari2022} (8th grade) & \bad{94.0\%} & 0.109 & 0.030 & \bad{94.6\%} & 0.127 & 0.034 & \good{97.2\%} & 0.132 & 0.031 & \good{94.9\%} & 0.071 & \better{0.018} \\
\citet{Akhtari2022} (headmaster) & \good{94.8\%} & 0.059 & 0.016 & \good{95.1\%} & 0.069 & 0.018 & \good{97.2\%} & 0.072 & 0.017 & \good{95.9\%} & 0.040 & \better{0.010} \\
\bottomrule
\end{tabular}
\end{adjustbox}
\caption{Empirical coverage (cvrg), average width, and root mean squared error (RMSE) of 95\% confidence intervals from different methods,
on simulated RDDs calibrated to twelve applications. We generate our
simulation distributions by training a Wasserstein Generative Adversarial Neural Network (WGAN) on the
original data \citep{athey2021using}. We report results for conventional local linear regression with Imbens--Kalyanamaran bandwidth \citep{imbens2012optimal}; robust bias-corrected (\texttt{rdrobust}) confidence intervals \citep{CCT2014robustCI}; bias-aware (\texttt{RDHonest}) confidence intervals \citep{ArmstrongKolesar2018}; and our proposed \texttt{plrd} method. We highlight (in \good{green}) the entries whose coverage is not statistically
significantly below nominal 95\% coverage (at the 5\% level using a one-sided binomial test) and (in \better{blue})
the entries with the lowest RMSE. All results are aggregated across
10,000 simulation replications.}
\label{tab:wgan}
\end{table}

\begin{figure}[t]
    \includegraphics[width=0.5\columnwidth]{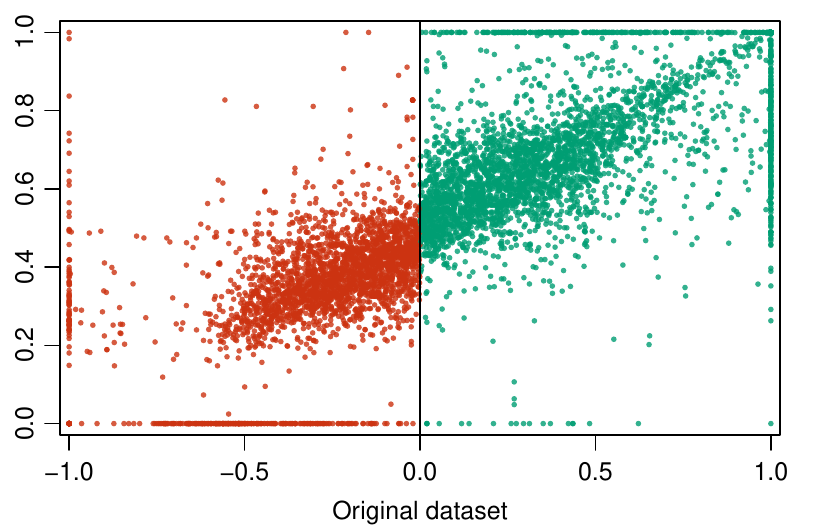}\includegraphics[width=0.5\columnwidth]{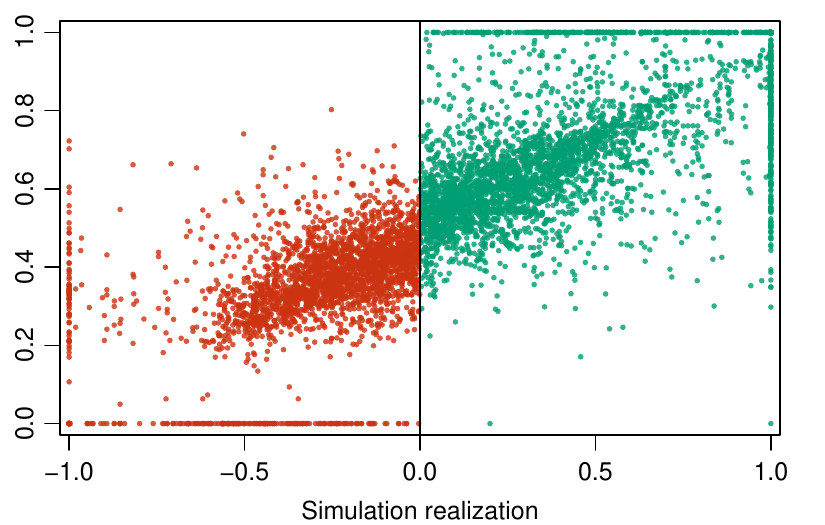}
    \caption{The original data from the Lee(2008) dataset \citep{lee2008}, as well as a simulated dataset produced using a WGAN trained to the original data \citep{athey2021using}.} 
    \label{fig:gan-lee}
\end{figure}

We summarize results for 95\% confidence intervals
in \cref{tab:wgan}; see \cref{sec:empirical-assessment} for details on the
simulation design. A first observation is that the conventional local linear regression intervals
systematically undercover the true parameter. 
The method \texttt{rdrobust}\footnote{The method \texttt{rdrobust} offers two options for building confidence intervals: A basic bias-corrected method which inherits the standard error estimator from the original local linear regression (and is justified by asymptotic arguments), and a robust method that incorporates a finite-sample variance inflation to account for the bias correction. The latter is the recommended one in the paper and this is the one we consider here.}
achieves better---although not perfect---coverage, with
nominal rates in {5 out of the 10 settings}. However, it does  so at the expense of wider confidence intervals than conventional local linear regression (18\% to 36\% wider in our applications).
\texttt{RDHonest}\footnote{In principle, the
method \texttt{RDHonest} asks the user to specify a worst-case bound on the curvature of the
underlying conditional response functions up front. However, the method also has a default
fall-back option where curvature bounds are derived using a provided rule-of-thumb
\citep{ArmstrongKolesar2020}; in our experiments we use this default.}
in turn achieves reliable coverage throughout all designs, but at the cost of still-wider
intervals: The median ratio between the width of these intervals and the \texttt{rdrobust} 
ones varies between 1.06 to 1.76 in our simulations, and between 1.25 to 2.11 relative to the conventional ones. 
Furthermore, these confidence intervals are often excessively conservative,
thus suggesting that some of this additional width is not strictly needed.\footnote{In several settings,
\texttt{RDHonest} achieves 98\% coverage or above---and this may come at a substantial cost in width.
For example, in the classical setting with no bias adjustment, Gaussian 98\% confidence intervals
are 18.7\% wider than the conventional 95\% intervals.}
These qualitative findings do not appear to be specific to our simulation designs; for example, in \cref{tab:simulations} we observe similar coverage properties when applying the method to pure noise
with no relationship between the outcome and the running variable.

\begin{figure}[t]
    \centering
    \includegraphics[width=0.49\linewidth]{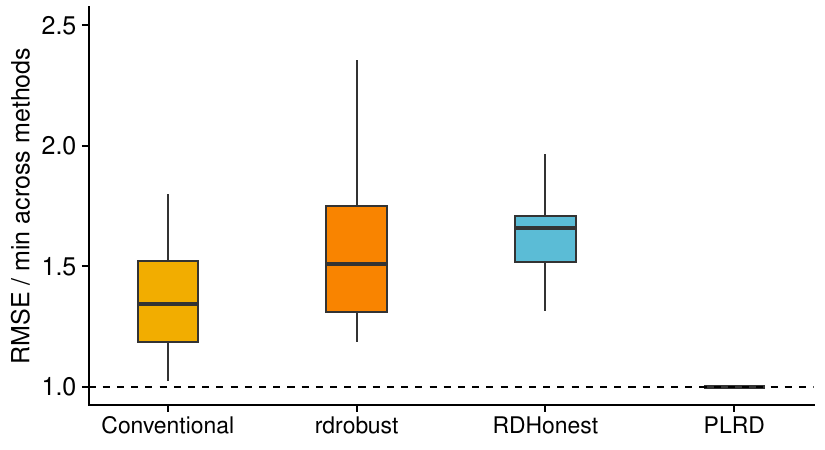}
    \includegraphics[width=0.49\linewidth]{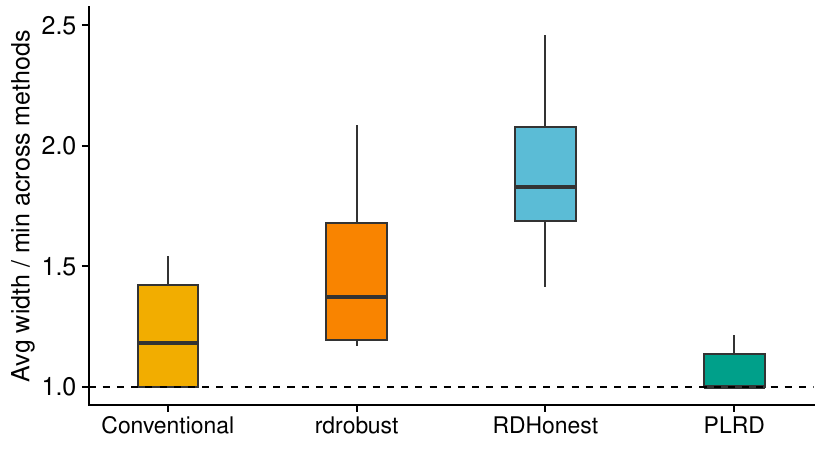}
    \caption{The relative RMSE and average width of 95\% confidence intervals from different methods as in \cref{tab:wgan}, relative to the minimum across the methods for each empirical application.}
    \label{fig:relative}
\end{figure}

These semi-synthetic simulation findings leave open a pragmatic question: Is it possible to design automated confidence intervals that reliably achieve nominal coverage
in realistic settings, while remaining competitive with the conventional intervals in terms of interval width? We emphasize here the ``realistic'' qualifier. Even if it is not feasible to achieve uniformly better properties, it may be that most real-world applications are sufficiently smooth to allow for nominal coverage and narrower confidence intervals than the existing methods.  

The main contribution of this paper is a method, partially linear regression discontinuity
inference (PLRD), that provides an affirmative answer to this question---at least in the context
of the numerical settings we consider. As described further below, algorithmically, the PLRD estimator
adapts \texttt{RDHonest} \citep{ArmstrongKolesar2018} to accommodate a partially linear treatment effect
model; conceptually, however, our full approach fits within an asymptotic econometric
framework as used for \texttt{rdrobust} \citep{CCT2014robustCI}. We implement our method in the package \texttt{plrd} for
\texttt{R} \citep{r2023}, which provides end-to-end functionality for regression discontinuity inference
without relying on the user to specify hard-to-choose tuning parameters. As illustrated in \cref{tab:wgan,tab:simulations}, our proposed \texttt{plrd} method achieves nominal and near-exact coverage throughout. It also has the most accurate point estimates across all  settings considered here (\cref{fig:relative}).

We emphasize that our semi-synthetic simulation setup, central to argument, comes with two qualifications. First, it relies on the availability of datasets where regression discontinuity methods are applicable. We apply it to twelve such data sets, but we encourage researchers to apply it to additional ones. Second, it relies on our generative adversarial nets yielding realistic data generating processes given the provided datasets. Well-known results in non-parametric inference imply that, fixing a sample size and given any dataset where standard regression-discontinuity estimators work well, one can construct another adversarial dataset that is statistically indistinguishable from the first but where the same estimators fail to get coverage \citep{low1997nonparametric}. Our approach to evaluation implicitly relies on an assumption that the real datasets we are calibrated to aren't adversarial in this sense.
Again, we encourage researchers to explore the sensitivity of our results to other generative models that are tied to existing data sets.
Finally, our estimation process involves some choices of tuning parameters for which different researchers may make different choices; we seek to mitigate concerns of researcher bias by using a single set of tuning parameters across all datasets.

\section{The PLRD Estimator}\label{sec:introduce-PLRD}

Consider the sharp RDD formalized using
potential outcomes \citep{IL2008review}. For each unit $i = 1, \, \ldots, \, n$, we
observe IID-sampled pairs $(X_i,\, Y_i)$, where $X_i \in \R$ is the running variable and
$Y_i \in \R$ is the outcome of interest. Following the sharp discontinuity design, we assume
there exists a cutoff $c \in \R$ such that treatment is assigned as $W_i = \ind{X_i\ge c}$.
We posit potential outcomes $\{Y_i(0), \, Y_i(1)\}$ such that $Y_i = Y_i(W_i)$, and our
goal is to estimate the conditional average treatment effect at the cutoff,
$\tau_c:=\mu_{1}(c)-\mu_{0}(c)$ where $\mu_{w}(x) := \E[Y_i(w)|X_i = x]$.

It is well known that, provided the $\mu_{w}(x)$ are continuous in $x$ and that the
running variable has continuous support around $c$, our target $\tau_c$ is identified
as a discontinuity in the conditional response surface at $X_i = c$ \citep{Hahn-et-al-2001}:
\begin{equation}\label{eqn:tau-definition}
\tau_c = \lim_{x\downarrow c} \mathbb{E}[Y_i \,|\, X_i = x] -\lim_{x\uparrow c}\mathbb{E}[Y_i \,|\, X_i = x].
\end{equation}
In order to obtain practical estimation or inference results for $\tau_c$, though, one
needs to make further regularity assumptions on the $\mu_{w}(x)$; our PLRD method relies
on the following two assumptions.

\begin{assumption}
\label{assu:smooth}
 $\mu_w(x)$ are 3 times continuously differentiable in $x$.
\end{assumption}

\begin{assumption}
\label{assu:part_lin}
The conditional average treatment effect function $\tau(x) = \mu_{1}(x)-\mu_{0}(x)$
has a linear dependence in $x$, i.e., $\tau(x) = \tau_c + \beta (x - c)$ for some
$\beta \in \R$.
\end{assumption}

\cref{assu:smooth} is a standard smoothness assumption; this is exactly the
same as the main smoothness assumption used by \citep{CCT2014robustCI} to justify
\texttt{rdrobust}. \cref{assu:part_lin} implies that the treatment effect
has a simpler dependence on the running variable than the baseline effect $\mu_{0}(x)$.
This assumption is well in line with assumptions often made in the literature on
heterogeneous treatment effect estimation \citep[e.g.,][Chapter 4]{wager2024causal};
however, it has not typically been used in the RDD literature to date, which is why
we emphasize this assumption in naming our method. Given these two assumptions, PLRD
proceeds as follows:
\begin{enumerate}\itemsep0mm
\item\label{step(1)} Under \cref{assu:smooth} there exists a finite constant $B$
such that the second derivatives $\mu''_w(x)$ are $B$-Lipschitz in a neighborhood
of $c$; and it is possible to {identify}  such a $B$ from data in the large-sample limit.
PLRD starts by estimating
$B$ via a cubic regression \citep{imbens2012optimal,CCT2014robustCI}.\footnote{We use sample splitting \cref{algo:PLRD} to
justify technical arguments underlying formal results given in \cref{sec:formal_considerations}. It
is possible, however, that an alternative argument could be used to justify similar results without
sample splitting; see, e.g., \citep[Appendix E]{ArmstrongKolesar2020} for results of this type.}
\item\label{step(2)} We form the set of all possible choices of $\mu_{0}(x)$ and $\mu_{1}(x)$ consistent
with the $B$-Lipschitz curvature bound estimated above and \cref{assu:part_lin}. Then,
we use numerical convex optimization
to derive a minimax linear estimator\footnote{Any estimator of the form
\smash{$\hat{\tau} = \scalebox{0.7}{$\sum$}_{i = 1}^n \gamma_i Y_i$} where the weights $\gamma_i$ only depend
on the realization of the running variables $X_i$ is called a ``linear'' estimator;
the minimax linear estimator is the estimator of this form with the minimax mean-squared
error conditionally on the $X_i$. The minimax linear estimation approach was introduced
by \citep{Donoho1994}, and first applied to RDDs by \citep{ArmstrongKolesar2018}. For a
textbook overview of the minimax linear estimation approach to inference for RDDs, see
\citep[Chapter 8.2]{wager2024causal}.}
for $\tau_c$ within this class \citep{ArmstrongKolesar2018,optrdd2019}. 
\item\label{step(3)} We identify the worst-case bias of our estimator within this class, and construct
``bias aware'' confidence intervals for $\tau_c$ using a construction
that simulatenously accounts for the standard error and
potential bias of an estimator \citep{ImbensManski2004}.
\end{enumerate}
Our full procedure---including some extra algorithmic details such as cross fitting---is
detailed as \cref{algo:PLRD}. As a pragmatic safeguard against failures of
\cref{assu:part_lin}, we start our procedure with a hypothesis test, and if we reject
this assumption with high confidence then we proceed with a variant of our method that
only relies on \cref{assu:smooth}; details of this flexible fallback are given in the Supplementary Material.

\begin{algorithm}[p]
\small
\setlength{\abovedisplayskip}{3pt}
\setlength{\belowdisplayskip}{3pt}
\smallskip

    \caption{Partially Linear Regression Discontinuity Inference (PLRD)}\label{algo:PLRD}
    This algorithm provides asymptotically valid $100(1-\alpha)\%$ confidence intervals for the conditional average treatment effect $\tau_c$ in a sharp RDD. This procedure is implemented in the \texttt{R} package \texttt{plrd}.

\smallskip

    \textbf{Input:} Data $\{(X_i,Y_i)\}_{i=1}^n$, the cutoff $c$, and the confidence level $1-\alpha$ (default is set as $95\%$).

\smallskip

    \textbf{Tuning Parameters:} Window size $\ell$ for estimating the smoothness parameter $B$ (default is set as the whole range), a lower bound $\varepsilon$ for $\wh B$ (default is $\varepsilon = {\operatorname{SD}[Y]} \, / \, 100$), and a confidence level $1-\alpha'$ for curvature detection (default is set as $99.9\%$). 

    \smallskip

    \textbf{Output:} The PLRD estimator $\wh{\tau}_{\texttt{PLRD}}$ with associated $1-\alpha$ confidence interval $\mc{I}_\alpha$.

    \smallskip

    \begin{algorithmic}[1] 
    \setlength{\itemsep}{1pt}
\State Discard all observations with $|X_i - c| > \ell$.
    \State Perform an ANOVA test between cubic polynomials fitted with or without different curvatures at significance level $1-\alpha'$ to detect whether different curvatures should be used. \label{anovatest}
    \State Randomly split the data into two near-equally sized folds, with index sets $I_{1}$ and $I_{2}$. 
        \State Estimate $\sigma^2$ by running the ordinary least-squares regression of $Y_i$ on the interaction of $X_i$ and $W_i$ for $i\in I_2$. Denote this estimate as $\wh{\sigma}^2_{(2)}$ and record the residuals from this regression as:
    \begin{equation}\label{eq:ri}
        r_i:=Y_i-\wh\mu_{W_i}(X_i)\text{ for each }i\in I_2.
    \end{equation}
\If{the test for different curvatures in Step~\ref{anovatest} is \emph{not} rejected}

  \State\label{first-within-if-loop}{Use the observations in fold $I_{2}$ to estimate the Lipschitz constant $B$ by fitting cubic polynomials with}
  \Statex \qquad   the same curvature as in~\eqref{eq:poly}, and denote this estimate as $\wh{B}_{(2)}$.
  
  \State Enforce a lower bound $\wh{B}_{(2)} \leftarrow \max\{\wh{B}_{(2)}, \, \varepsilon \}$.

  \State\label{last-within-if-loop}Let $\rhos:=\{f: f(0)=f'(0)=f''(0)=0, f'' \text{ is } \mtd\text{-Lipschitz}\}$, and solve the following optimization problem
  \Statex \qquad  via quadratic programming (see the Supplementary Material for implementation details):
    \begin{equation}
    \label{eq:point_est}
    (\wh\g_{(1)},\wh{t}_{(1)}):= \mmz_{\gamma,\,t} \left(\wh{B}^2_{(2)}t^2 +\wh{\sigma}^2_{(2)} \sum_{i\in I_{1}} \g_i^2\right)\quad \text{subject to}\quad \sup_{\rho\in \rhos} \ \sum_{i\in I_{1}} \g_i \rho(X_i-c)\le t.
    \end{equation}
    \State Repeat Steps \ref{first-within-if-loop} to \ref{last-within-if-loop} with the roles of $I_1$ and $I_2$ swapped.
\Else
  \State\label{first-within-else}{Use the observations in fold $I_{2}$ to estimate the Lipschitz constant $B$ by fitting cubic polynomials with}
  
  \Statex \qquad  different curvatures (see the Supplementary Material), and denote this estimate as $\wh{B}_{(2)}$.

  \State Enforce a lower bound $\wh{B}_{(2)} \leftarrow \max\{\wh{B}_{(2)}, \, \varepsilon \}$.

  \State\label{last-within-else}Let $\rhos:=\{f: f(0)=f'(0)=f''(0)=0, f'' \text{ is } \mtd\text{-Lipschitz}\}$, and solve the following optimization problem
  \Statex \qquad  via quadratic programming (see the Supplementary Material for implementation details):
    \begin{equation}
    \label{eq:point_est_diffcurv}
    (\wh\g_{(1)},\wh{t}_{(1)}):= \mmz_{\gamma,\,t} \left(\wh{B}^2_{(2)}t^2 +\wh{\sigma}^2_{(2)} \sum_{i\in I_{1}} \g_i^2\right)\quad \text{subject to}\quad \sup_{\rho_0,\rho_1\in \rhos} \ \sum_{i\in I_{1}} \g_i \rho_{W_i}(X_i-c)\le t.
    \end{equation}
    \State Repeat Steps \ref{first-within-else} to \ref{last-within-else} with the roles of $I_1$ and $I_2$ swapped.
\EndIf

    \State Set $\wh\g_i = \g_{(k),i}/2$ if $i\in I_k$ ($k=1,2$), and $\wh\g_i=0$ if $|X_i-c|>\ell$.
    Compute
    $$\wh\tau_{\texttt{PLRD}}=\sum_{i=1}^n \wh\g_i Y_i \ \  \text{ and }  \ \ \wh{s}^2(\wh\g) = \sum_{i=1}^n \wh\g_i^2r_i^2,$$ where the $r_i$ are the residuals from the least squares regressions in \eqref{eq:ri}.
    
    \State Report the following $100(1-\alpha)\%$ confidence interval for $\t$: 
    \begin{equation}
    \label{eq:alg_CI}
    \cal{I}_\alpha := [\wh{\tau}_{\texttt{PLRD}}-\wh{h}_\alpha,\,\wh\tau_{\texttt{PLRD}}+\wh{h}_\alpha],\quad \wh{h}_\alpha:= \min \,\{h: \P(|b+\wh{s}(\wh\g) Z| \leq h) \geq 1-\alpha \text { for all }|b| \leq \wh{b}\,\},
    \end{equation}
    where $Z \sim \mathcal{N}(0,1)$, and $\wh{b}:=\big(\wh{B}_{(2)}\wh{t}_{(1)}+\wh{B}_{(1)}\wh{t}_{(2)}\big)/2$.
    \end{algorithmic}
\end{algorithm}

Our main formal result about the PLRD estimator is that it yields
valid confidence intervals under our assumptions in the large-sample
limit. 

\begin{theorem}\label{thm:validity}
Consider a sharp regression discontinuity design under \cref{assu:smooth,assu:part_lin},
where the running variable $X_i$ has a strictly positive density at the threshold $c$, 
$\sigma_i^2=\var(Y_i\mid X_i)$ is bounded away from $0$, i.e., $\sigma_i\ge \sigma_{\min}$ for some $\sigma_{\min}>0$, and $\E[|Y_i - \mu_{W_i}(X_i)|^q \mid X_i = x] \le C$ uniformly over all $x$, for some $q > 2$ and $C\in (0,\infty)$.
Suppose furthermore that we execute the initial cubic regression used to pick
$B$ over a pilot bandwidth of width $\ell$ with
\begin{equation}
\label{eq:ell_decay}
\lim_{n \rightarrow \infty}   n^{1/12}\ell = 0, \ \ \ \ \lim_{n \rightarrow \infty}   n^{1/7} \ell= +\infty,
\end{equation}
and choose any $\varepsilon > 0$ in Algorithm \ref{algo:PLRD}.
Then, level-$(1 - \alpha)$ PLRD intervals $\mathcal{I}_\alpha$ are asymptotically valid: For any $0 < \alpha < 1$,
\begin{equation}
\label{eqn:asymp-valid}
    \liminf_{n\to\infty} \, \P[\tau_c \in \mathcal{I}_\alpha] \ge 1 - \alpha.
\end{equation}
\end{theorem}

We note that our algorithm does in principle require two tuning parameters,
namely the pilot bandwidth $\ell$ used to run the initial cubic regression to choose $B$
and the lower bound $\varepsilon$ for $\wh B$. We note, however, that our algorithmic performance
is largely insensitive to these choices (and simple defaults appear to provide good performance
in practice); and similar objects are also required for other methods \citep{CCT2014robustCI,imbens2012optimal}.
First, regarding $\ell$, this bandwidth can in practice be chosen simply to be
``large'': We only need $\ell$ to decay faster than $n^{-1/12}$, whereas the mean-squared
error optimal bandwidth for local linear regression in our setting (i.e., under \cref{assu:smooth}) scales as $n^{-1/7}$ \citep{Cheng-Fan-Marron}.
In our implementation, we by default simply set $\ell$ to contain the full range of the data, and this choice is what we use in all our experiments\footnote{This choice mirrors \citep[p.~942, {\it step 2}]{imbens2012optimal} who seed their bandwidth-selection procedures by fitting global cubic polynomials to the full dataset; see also \citep[S.2.6, {\it step 0}]{CCT2014robustCI}.} (\cref{tab:wgan}). 
As such, practitioners should be able to get reliable results out of our method unless the nature of the relationship between $X_i$ and $Y_i$ changes completely as we get far from the cutoff---and in settings like this it would be possible to pick $\ell$ by plotting the data. 
Meanwhile, the lower bound $\eps$ for $\wh{B}$ will have no effect in large samples as long as we choose it to be small enough that $\varepsilon < |\mu'''(c)|$.\footnote{One interesting
question is what happens at the critical point where $\mu'''(c) = 0$; as, in this case,
our constraint does necessarily bind and we fall back to using $\varepsilon$ as our curvature
bound. The challenge is that, when $\mu'''(c) = 0$, our cubic polynomial estimator \smash{$\hat B$} (in step 6 of \cref{algo:PLRD})
converges to 0, but the exact way it does so is left underspecified under \cref{assu:smooth}.
Our use of the lower bound $\varepsilon$ guarantees validity of our confidence intervals at
the expense of some excess conservativeness when $\mu'''(c) = 0$. We note that existing
methods, including \citep{CCT2014robustCI} and \citep{imbens2012optimal}, also do not
offer guidance on how to potentially benefit from ``superefficiency'' in
similar critical cases; i.e., our approach is again in line with current best practices.}

\subsection{Discrete running variables}

    Discrete running variables arise frequently in regression discontinuity applications. 
    Our formal results are stated under the assumption that the running variable has a continuous density bounded away from zero near the cutoff, which does not accommodate the discrete case. In particular, when the running variable is discrete, {\it Step \ref{step(1)}} is no longer guaranteed to consistently recover $B$ under \cref{assu:smooth} alone, and thus \cref{thm:validity} does not provide a formal guarantee for the discrete case. However, {\it Steps \ref{step(2)}} and {\it \ref{step(3)}} of our procedure remain valid, i.e., given an appropriate $B$ they yield partial identification intervals for $\tau_c$ with a guarantee of the type \eqref{eqn:asymp-valid} conditional on that bound \citep{KolesarRothe2018,optrdd2019}.
    Thus, PLRD is robust to discrete running variables in the following sense: The interval logic \eqref{eqn:asymp-valid} is sound as long as the practitioner is willing to commit to a curvature bound---whether from the pilot regression or from background knowledge. 

\cref{tab:wgan} illustrates the empirical performance of PLRD when the running variable is discrete: Five of our twelve benchmark cases---Matsudaira (math and reading), Jacob \& Lefgren (math and reading), and Oreopoulos---use assignment rules based on test-scores, where the running variable takes values on a finite grid. This suggests that the pilot-based curvature estimate continues to behave reasonably even in the case of discrete running variables, where one would need further assumptions to formally justify the data-driven detection of $B$ in \cref{algo:PLRD}.

\section{Relationship to Existing Methods}\label{sec:relationship-to-existing-methods}

Regression discontinuity designs were originally introduced by \citep{thistlethwaite1960regression}.
The modern econometric literature goes back to \citep{Hahn-et-al-2001}, who emphasized the identification
result \eqref{eqn:tau-definition} and recommended estimation $\tau_c$ via local linear regression,
\smash{$\wh\tau_\texttt{LLR}$}.
In the following decade, local linear
(or polynomial) regression developed into the prevalent approach to inference in RDDs, with confidence
intervals produced by applying off-the-shelf tools for heteroskedasticity-robust inference directly to the
regression \citep{IL2008review,lee2010regression}.

A limitation of this approach, however,
is that it leads to contradictory guidance on the bandwidth choice under a non-parametric
specification as given in, {\it e.g.},  \cref{assu:smooth}. The accuracy of
\smash{$\wh\tau_\texttt{LLR}$} depends delicately on the bandwidth $h_n$, with a first systematic proposal given in \citep{imbens2012optimal}.
For choices of $h_n$ that optimize this accuracy, the bias and
standard error of \smash{$\wh\tau_\texttt{LLR}$} are of the same order, so off-the-shelf
local linear regression inference will not have valid coverage for this choice of $h_n$.
The practitioner is thus forced to invoke ``undersmoothing'', whereby they voluntarily
use a non-accuracy-optimizing choice of $h_n$ to justify validity of their inferential
procedure.

In a major advance, \citep{CCT2014robustCI} and
\citep{ArmstrongKolesar2018} introduced paradigms for inference in RDDs that formally
account for the bias of local linear regression.
These methods offer material improvements in the credibility of regression-discontinuity
analyses over the earlier practice \citep{stommes2023reliability}. The main idea in \citep{CCT2014robustCI} is to
estimate and correct for the bias of local linear regression by leveraging higher-order smoothness;
we refer to this approach as the ``bias-corrected'' approach;\footnote{When
reporting results using the bias-correct approach, we use robust standard error
estimates as recommended by \citep{CCT2014robustCI} and as implemented in the package
\texttt{rdrobust} \citep{calonico2015rdrobust,calonico2017rdrobust}.}
this approach as further been extended in a number of directions \citep{CCFT2019,cattaneo2022regression}.
In contrast, the ``bias-aware'' approach developed by
\citep{ArmstrongKolesar2018} involves conservatively widening
intervals to account for worst-case bias (without needing to actually estimate the bias);
this approach has also been recently developed further
\citep{KolesarRothe2018,optrdd2019,noack2024bias}.

Interestingly, not only do the above papers under the bias-corrected vs.~bias-aware paradigms
explore different methodological strategies for inference in RDDs; they also focus on proving
formal results under different formal paradigms. \citep{CCT2014robustCI} and follow-ups focus
entirely on asymptotic validity results---as we do here. They fix a data-generating process, and
guarantee that they eventually get coverage for this data-generating process as the sample
size gets large (as is also done in \eqref{eqn:asymp-valid}). In contrast, most papers written
under the bias-aware paradigm specify a quantitative smoothness class for the $\mu_w(x)$
({\it e.g.}, all functions with a second derivative bounded by a pre-specified constant $B$), and
then seek guarantees that hold in finite samples and uniformly over this whole class.
Getting uniform guarantees in finite samples of course requires stronger assumptions
than getting pointwise asymptotic guarantees; and for this reason papers associated with the
bias-aware paradigm often start by making stronger assumptions than those associated with the
bias-corrected paradigm. We emphasize, however, that this is due to the sought formal
guarantees; and it's not that bias-aware methods inherently require stronger assumptions to work
than bias-corrected methods.

Our proposed method, PLRD, straddles these two literatures. Our confidence intervals
are rooted in the bias-aware paradigm \citep{ArmstrongKolesar2018}, and
we use numerical convex optimization to construct our estimator \citep{optrdd2019}.
In contrast, our formal results are focused on pointwise asymptotics, and
the assumptions and guarantee type given in our \cref{thm:validity} mirror
those in Theorem 1 of \citep{CCT2014robustCI}. One notable paper that also straddles
these literatures is \citep{ArmstrongKolesar2020}, who propose bias-aware confidence intervals
for local linear regression with pointwise asymptotic guarantees. In particular, they show
that if one runs local linear regression with a mean-square-error
optimal bandwidth choice, then under generic
conditions we can build valid, bias-aware 95\% confidence intervals for $\tau_c$ as
$\wh\tau_\texttt{LLR} \pm 2.18 \times \text{standard errors}$ rather than the conventional
$\wh\tau_\texttt{LLR} \pm 1.96 \times \text{standard errors}$. We provide a comparison of
PLRD to this approach in the Supplementary Material.

Wrapping up, PLRD can best be understood within the context of earlier developments in
the bias-aware paradigm; and in fact {\it Steps (\ref{step(2)})} and {\it (\ref{step(3)})} of our algorithmic outline fall
squarely within the roadmap for minimax linear inference via numerical convex optimization
as in \citep{optrdd2019}. However, in a deviation from existing work:
\begin{itemize}\itemsep0mm
    \item We rely on higher-order smoothness and a preliminary cubic fit to initialize
    tuning parameters of our algorithm; in \citep{ArmstrongKolesar2018}, \citep{KolesarRothe2018} and \citep{optrdd2019}, a parameter analogous to the $B$ derived in our {\it Step (\ref{step(1)})} would
    need to be pre-specified using subject matter knowledge or chosen using a heuristic.
    \item Because of this data-driven choice of $B$ in {\it Step (\ref{step(1)})}, we are not able to provide
    uniform, finite-sample inferential guarantees;\footnote{As shown in \citep{ArmstrongKolesar2018},
    doing so would in fact be impossible under smoothness assumptions like \cref{assu:smooth} alone.
    We do note that proving uniform bounds for methods involving data-driven smoothness guarantees
    can be possible under more restricted function classes, e.g., \citep[Appendix E]{ArmstrongKolesar2020}
    consider such results over function classes where global polynomial approximations are guaranteed
    to capture worst-case local curvature of the underlying functions. We leave a study
    of uniformity guarantees for PLRD over restricted function classes to future work.}
    thus, like in \citep{CCT2014robustCI},
    we focus on pointwise asymptotic inferential statements.
    \item The specific worst-case class we use in our {\it Step (\ref{step(2)})}, i.e., with $B$-Lipschitz
    curvature bounds and a linear treatment effect assumption, has not been previously
    considered in the literature. The most standard choice in applications of
    \citep{ArmstrongKolesar2018}, \citep{KolesarRothe2018} and/or \citep{optrdd2019}
    has been to only assume that the $\mu_w(x)$ have second derivatives bounded by $M$.
\end{itemize}
Regarding this last point, the earlier focus on the $M$-bounded curvature class appears
to have been largely guided by intuitive simplicity of this class rather than any
conceptual or empirical arguments. Given the empirical performance of PLRD on calibrated
simulations, we suggest that our $B$-Lipschitz curvature class with linear treatment effects
may be a good default choice in many applications. 

To further examine the effect of different choices on the behavior of our estimator, we conduct an ablation analysis, reported in the Supplementary Material,
comparing four alternatives to PLRD: Bias-aware intervals for local linear regression
with the inflated critical value of \citep{ArmstrongKolesar2020}; the analogous construction for local quadratic regression; PLRD with \cref{assu:part_lin} removed while retaining the $B$-Lipschitz curvature structure (so that separate curvature bounds are allowed on each side of the threshold); and PLRD with \cref{assu:smooth} relaxed to second-order smoothness, which essentially recovers the optrdd construction of \citep{optrdd2019} augmented with \cref{assu:part_lin}. The ablation analysis empirically demonstrates that the efficiency of PLRD is a joint consequence of both \cref{assu:part_lin,assu:smooth}. 


\section{Empirical Study}\label{sec:empirical-assessment}

We now describe in more detail the experiment presented in the introduction (\cref{tab:wgan}).
As context for this experiment, \cref{tab:real-data-CIs} shows a replication
of a number of high-profile RDD applications \citep{lee2008,cattaneo2015senate,LudwigMiller,meyersson2014islamic,matsudaira,jacob2004,Oreopoulos,Akhtari2022} using a variety of methods;
further description of each benchmark dataset is provided in the Supplementary Material.

\begin{table}[t]
\small
\centering\renewcommand{\arraystretch}{1}
\begin{adjustbox}{max width=\textwidth}
\begin{tabular}{@{}lcccc@{}}
\toprule
\textbf{Dataset} & \multicolumn{1}{c}{conventional} & \multicolumn{1}{c}{\texttt{rdrobust}} & \multicolumn{1}{c}{\texttt{RDHonest}} & \multicolumn{1}{c}{\texttt{plrd}}\\
\midrule
\citet{lee2008} & $0.064 \pm 0.023$ & $0.059 \pm 0.026$ & $0.058 \pm 0.032$ & $0.073 \pm 0.024$\\
\citet{cattaneo2015senate} & $7.416 \pm 2.862$ & $7.504 \pm 3.419$ & $8.030 \pm 4.210$ & $5.830 \pm 2.127$\\
\citet{LudwigMiller} & $-2.440 \pm 2.238$ & $-2.809 \pm 2.535$ & $-3.153 \pm 2.828$ & $-1.249 \pm 1.507$\\
\citet{meyersson2014islamic} & $3.019 \pm 2.782$ & $2.979 \pm 3.284$ & $3.024 \pm 3.109$ & $3.010 \pm 2.093$\\
\citet{matsudaira} (reading) & $0.036 \pm 0.029$ & $0.030 \pm 0.033$ & $0.031 \pm 0.046$ & $0.017 \pm 0.033$\\
\citet{matsudaira} (math) & $0.075 \pm 0.025$ & $0.072 \pm 0.030$ & $0.057 \pm 0.041$ & $0.073 \pm 0.028$\\
\citet{jacob2004} (reading) & $0.130 \pm 0.042$ & $0.143 \pm 0.048$ & $0.127 \pm 0.046$ & $0.100 \pm 0.040$\\
\citet{jacob2004} (math) & $0.195 \pm 0.035$ & $0.203 \pm 0.039$ & $0.187 \pm 0.058$ & $0.168 \pm 0.039$\\
\citet{Oreopoulos} & $0.070 \pm 0.104$ & $0.099 \pm 0.143$ & $0.070 \pm 0.125$ & $0.021 \pm 0.064$\\
\citet{Akhtari2022} (4th grade) & $0.022 \pm 0.047$ & $0.018 \pm 0.056$ & $-0.020 \pm 0.049$ & $0.042 \pm 0.046$\\
\citet{Akhtari2022} (8th grade) & $0.040 \pm 0.058$ & $0.049 \pm 0.068$ & $0.020 \pm 0.068$ & $0.036 \pm 0.054$\\
\citet{Akhtari2022} (headmaster) & $0.274 \pm 0.035$ & $0.267 \pm 0.040$ & $0.269 \pm 0.043$ & $0.270 \pm 0.033$\\
\bottomrule
\end{tabular}
\end{adjustbox}
\caption{95\% confidence intervals produced by different methods on cited benchmark datasets.}
\label{tab:real-data-CIs}
\end{table}

A fundamental challenge in assessing methods for causal inference on real-world data is that we
typically do not have access to a ground truth we can use for evaluation. In light of this
challenge, \citep{athey2021using} proposed using Wasserstein Generative Adversarial Networks (WGANs)
to generate synthetic data that mimic benchmark datasets, such as to obtain simulation instances
that enable comparing methods in realistic settings.
Our WGAN evaluation pipeline involves the following three steps:
\begin{enumerate}
    \item \textbf{Training:} We train WGANs on each benchmark dataset to mimic the underlying data generating process 
    for both treated and control units.
    \item \textbf{Generation:} Once trained, the WGANs generate multiple synthetic datasets that reflect the characteristics of the original data. 
    \item \textbf{Evaluation:} We apply the various methods, including our proposed \texttt{PLRD} method, to the synthetic datasets.  The performance of these methods is then evaluated based on the empirical coverage and the average width of the confidence intervals. 
\end{enumerate}
As illustrated in Figure \ref{fig:gan-lee}, a trained WGAN is able to simulate data whose
distribution is very similar to that of the original data; thus suggesting that if a method
gets reliable coverage on the simulator we can also trust it on the real data.

Our WGAN training procedure consists of the following steps.  First,  apply marginal CDF transforms to the data to convert $(X_i, Y_i)$ to $(U_i, V_i)$, where $U_i=F_x(X_i)$ and $V_i = F_y(Y_i)$. For discrete data we break ties by randomization, and for additional stability we use a Gaussian transformation. We train two generators $\wh{g}_w(u, z)$ ($w=0,1$), one for treated units and one for controls, that mimic the conditional distribution of the CDF-transformed outcome $V$ given the CDF-transformed running variable $U$. Finally, we invert the CDF-transforms to obtain $(\wt{X}_i,\wt{Y}_i)$ from the pair $(\wt{U}_i, \wt{V}_i)$ generated by the trained WGAN. 

We calculate the ground truth $\widetilde{\tau}$ for the GAN estimated distribution as follows. We first calculate the cutoff $c'$ under the applied CDF transforms. We then sample random noise $Z_{0,q}, Z_{1,q}\in \R$ for $q=1,\dots,Q$, where $Q$ is a large number ($10^7$ in our simulations), generate CDF-transformed outcomes at the threshold $c'$, and invert them:
$$\widetilde{\tau}:=\frac{1}{Q}\sum_{q=1}^Q \left(F_y^{-1}(\wh{g}_1(c', Z_{1,q}))-F_y^{-1}(\wh{g}_0(c',Z_{0,q}))\right).$$
We generate $N=1,000$ simulated datasets indexed by $i=1,\dots,N$, and compute confidence intervals $[a_{i,k},b_{i,k}]$, for each method $k$. The empirical coverage for the $k$-th method is given by $\wh{p}_k=N^{-1}\sum_{i=1}^N \mathbf{1}\{a_{i,k}\le \widetilde{\tau}\le b_{i,k}\}$. \cref{tab:wgan} shows the obtained results.


\subsection{Revisiting existing simulation designs}\label{sec:simulations-from-CCT}


\begin{table}[t]
\small
\centering
\setlength{\tabcolsep}{2.7pt}
\renewcommand{\arraystretch}{1.0}
\begin{tabular}{@{}l*{12}{c}@{}}
\toprule
\multirow{2}{*}{}  
& \multicolumn{3}{c}{conventional}  
& \multicolumn{3}{c}{\texttt{rdrobust}} 
& \multicolumn{3}{c}{\texttt{RDHonest}} 
& \multicolumn{3}{c}{\texttt{plrd}}\\ 
\cmidrule(lr){2-4} \cmidrule(lr){5-7} \cmidrule(lr){8-10} \cmidrule(lr){11-13} 
\textbf{Dataset} 
& \textbf{cvrg} & \textbf{width} & \textbf{rmse}
& \textbf{cvrg} & \textbf{width} & \textbf{rmse}
& \textbf{cvrg} & \textbf{width} & \textbf{rmse}
& \textbf{cvrg} & \textbf{width} & \textbf{rmse} \\ 
\midrule
\textbf{Pure noise}
& \bad{92.3\%} & 1.392 & 0.398
& \bad{92.9\%} & 1.652 & 0.460
& \good{96.0\%} & 1.854 & 0.445
& \good{97.2\%} & 1.011 & \better{0.228} \\

\textbf{Setting 1}
& \bad{88.4\%} & 0.203 & 0.064
& \bad{90.9\%} & 0.240 & 0.071
& \good{95.4\%} & 0.288 & 0.069
& \good{94.8\%} & 0.237 & \better{0.061} \\

\textbf{Setting 2}
& \bad{79.7\%} & 0.266 & 0.099
& \bad{90.5\%} & 0.296 & 0.085
& \good{96.2\%} & 0.440 & 0.095
& \good{97.3\%} & 0.389 & \better{0.090} \\

\textbf{Setting 3}
& \bad{90.1\%} & 0.213 & 0.064
& \bad{92.7\%} & 0.246 & 0.069
& \good{96.2\%} & 0.301 & 0.070
& \good{99.6\%} & 0.290 & \better{0.050} \\

\textbf{Setting 4}
& \bad{91.8\%} & 0.201 & 0.059
& \bad{92.3\%} & 0.238 & 0.068
& \good{96.6\%} & 0.276 & 0.064
& \good{96.1\%} & 0.178 & \better{0.045} \\
\bottomrule
\end{tabular}
\caption{Empirical coverage, average width, and RMSE of 95\% confidence intervals produced by different methods on pure noise and on simulations from \citep{CCT2014robustCI}. All simulations have sample size $500$, and results are aggregated across $10{,}000$ replications. Methods are run as in \cref{tab:wgan}.}
\label{tab:simulations}
\end{table}

For completeness, we also benchmark the PLRD estimator against other existing methods on
some standard simulation designs. First, we consider a pure noise setting where we generate
the running variables as $X_i \sim \text{Unif}([-1, \, 1])$, the outcomes as $Y_i \sim \mathcal{N}(0,\, 1)$,
and set treatment to $W_i = 1(\{X_i \geq 0\})$.
Next, we consider simulation settings used in \citep[Section 6]{CCT2014robustCI}.
In each of these settings, we draw the running variable $X_i\sim 2 \, \text{Beta}(2, 4)-1$,
and the outcome is given  by
$Y_i = m(X_i)+ \eps_i$, where $\eps_i$ are mean-zero Gaussian noise with variance $0.1295^2$.
The form of $m(\cdot)$ varies for the different simulation settings;
see \citep{CCT2014robustCI} for details on Settings 1 through 3.
Setting 4 uses the function $m=m_\text{quad}$ as in \citep[Section 4.6]{imbens2012optimal}.
We note that the linear treatment effect condition (\cref{assu:part_lin}) is not satisfied
in several of these designs, so we expect our pragmatic safeguard (given in \cref{algo:PLRD})
to play an important role in delivering robust performance.

\cref{tab:simulations} presents the empirical coverage and average width
of 95\% CIs produced by different methods across these simulation settings.
As before, \texttt{rdrobust} falls short of nominal coverage; note that these
numbers are in line with those originally reported by \citep[Table 1]{CCT2014robustCI}.
\texttt{RDHonest} delivers valid coverage at the expense of wider intervals, and \texttt{plrd} delivers valid coverage with intervals
whose width is competitive with \texttt{rdrobust}.

\section{Formal Results}\label{sec:formal_considerations}

Our point estimator, i.e., the $\wh{\g}_i$-weighted average of outcomes $Y_i$ with weights as
given in \eqref{eq:point_est} is, ignoring cross-fitting, equivalent to a minimax linear
estimator for $\tau_c$ over the class \smash{$\cal{M}_{\wh{B}}$} of all functions satisfying  \cref{assu:part_lin} and for which
\smash{$\mu''_0(x)$} is globally $\wh{B}$-Lipschitz, i.e.,\footnote{The conditional variance of
a linear estimator as in \eqref{eqn:minimax-problem0} is
\smash{$\scalebox{0.7}{$\sum$}_{i = 1}^n \hat{\gamma}_i^2 \sigma_i^2$} with
\smash{$\sigma_i^2 = \operatorname{Var}[Y_i \,|\, X_i]$}. If we had
\smash{$\sigma_i^2 = \hat{\sigma}^2$} for all $i = 1, \, \ldots, \, n$, then
\eqref{eqn:minimax-problem0} would exactly be the minimax linear estimation
problem for $\tau_c$. In the case of heteroskedasticity, the term
\smash{$\hat{\sigma}^2 \lVert \gamma \rVert_2^2$} should be interpreted as a
variance proxy that yields a practical estimator---and one that still maintains
provable validity guarantees under heteroskedasticity; see \citep[p.~272]{optrdd2019}
for further discussion.} $\wh{\tau}_{\texttt{PLRD}} := \sum_{i=1}^n \wh{\gamma}_{i} Y_i$, where
\begin{equation}\label{eqn:minimax-problem0}
  \begin{aligned}
  &\wh{\g}
  :=\argmin_{\gamma}
  \left\{
  \wh{\sigma}^2 \lVert \gamma \rVert_2^2
  + \sup_{\mu_{0},\, \mu_{1}\in\,\cal{M}_{\wh{B}}}
  \Delta(\gamma;\mu_0,\mu_1)^2
  \right\},\\
  &\Delta(\gamma;\mu_0,\mu_1)
  :=
  \sum_{i=1}^n {\gamma}_{i} \mu_{W_i}(X_i)
  - \{\mu_1(c) - \mu_0(c)\}.
  \end{aligned}
\end{equation}
If $\wh{B}$ were a constant chosen a-priori and $\mu''_0(x)$
were in fact globally $\wh{B}$-Lipschitz, we could use existing results for minimax linear regression-discontinuity
inference \citep[e.g.,][]{ArmstrongKolesar2018,optrdd2019} to directly provide guarantees for $\wh{\tau}_{\mathtt{PLRD}}$.
But here, of course, this is not the case: Our choice of $\wh{B}$ is learned in {\it Step \ref{step(1)}} of our algorithm, and
\cref{assu:smooth} only implies local---not global---Lipschitz continuity in $\mu''_0(x)$ .

This section records the main technical statements underlying \cref{thm:validity}. 
We start by verifying validity of the cubic regression in {\it Step \ref{step(1)}} of our procedure. The following result is a direct consequence of standard results for local polynomial
regression \citep[e.g.,][Theorem 3.1]{fan1996framework}, and implies that
\smash{$\wh{B}_{(1)}$} and \smash{$\wh{B}_{(2)}$} in \cref{algo:PLRD} satisfy \smash{$\wh{B}_{(k)}\rightarrow_p B$}
as $n\to\infty$.

\begin{proposition}
\label{prop:B}
Suppose that \cref{assu:smooth,assu:part_lin} hold and that $X_i$ has a density bounded away
from 0 in a neighborhood of $c$, and $B := \max\{|\mu_0'''(c)|,|\mu_1'''(c)|,\eps\}$. 
Run the local cubic regression: 
\begin{equation}
\label{eq:poly}
\begin{aligned}
\wh\beta
&:=\argmin_{\beta\in \R^6}
\sum_{\substack{|X_i - c| \leq \ell}}
\bigg(
Y_i - \beta_{W_i,0}-\beta_{W_i,1} (X_i-c)
-\frac{\beta_{2}}{2}(X_i-c)^2
-\frac{\beta_{3}}{6}(X_i-c)^3
\bigg)^2,
\end{aligned}
\end{equation}
with bandwidth $\ell$ scaling as in \eqref{eq:ell_decay}, and set\footnote{As an additional safe-guard against small sample sizes, we calculate a standard error estimate $\mathrm{s.e.}(\hat{\beta}_3)$ and use the conservative yet asymptotically consistent estimate
$\hat B = \max\{|\hat\beta_{3}\pm 1.96\cdot\mathrm{s.e.}(\hat\beta_3)|,\eps\}$ in our software.}
$\hat B = \max\{|\beta_{3}|,\eps\}$.
Then, $\wh B\rightarrow_p B$ as $n\to\infty$.
\end{proposition}

Our core analytic result is the following central limit theorem for our PLRD estimator.
Asymptotic validity of the PLRD confidence intervals as claimed in \cref{thm:validity}
then follows immediately.
Proofs of \cref{thm:validity,thm:main-CLT}, together with the auxiliary lemmas, are given in the Supplementary Material.

\begin{theorem}
\label{thm:main-CLT}
Under the conditions of \cref{thm:validity}, $\wh{\tau}_{\mathtt{PLRD}}$ as constructed in \cref{algo:PLRD} satisfies 
\begin{equation}
\label{eq:clt_eq}
\begin{aligned}
{(\wh{\tau}_{\mathtt{PLRD}}-b(\wh{\gamma})-\t)} \,\big/\,{s(\wh\g)}
&\dto \normal(0, \, 1),
\end{aligned}
\end{equation}
where $b(\wh{\g})=\sum_{i=1}^n \wh{\g}_{i}\, \mu_{W_i}(X_i)-\t$ is the conditional bias of $\wh{\tau}_{\mathtt{PLRD}}$, and ${s}^2(\wh\g)
= \sum_{i=1}^n\wh\gamma_{i}^2\sigma_i^2$ where $\sigma_i^2 =\var(Y_i\mid X_i)$.
Furthermore,
\begin{equation}
\label{eq:bias_ok}
\limsup_{n\to\infty} \left(\left| b(\wh{\g})\right| - \wh{b}\right) \,\big/\, {s}(\wh\g) \leq_p 0,
\end{equation}
i.e., the learned bias bound $\wh{b}$ in \eqref{eq:alg_CI} is conservative for the true bias;
and ${s}(\wh\g)=\O(n^{-3/7})$.
\end{theorem}

\section{Extensions}


It is often of interest to incorporate auxiliary covariates that are predictive of the outcome but unrelated to the treatment assignment. One approach to doing so is by augmenting \eqref{eqn:minimax-problem0} with covariate-balance conditions as linear constraints \citep{optrdd2019}, provided that exact balance is feasible. Another approach is to fit a flexible predictive model and replace the outcomes with the residuals from this fit \citep{noack2021flexible}.
Investigation of these approaches is left for future work.


Extending PLRD to fuzzy RDDs (
see~\citep{IL2008review} for a review) requires an analogue of \cref{assu:part_lin} for the local average treatment effect. 
A linear shape restriction on this imposes a joint constraint on the numerator and the denominator, and is less transparent than in the sharp-design. 
An alternative 
formulation would instead impose partial linearity on the intent-to-treat function and treat the first-stage discontinuity as a separate object subject to its own smoothness constraint, similar in spirit to \citep{noack2024bias}. 


Adapting PLRD to clustered data requires replacing the heteroskedasticity-robust sandwich variance in the last step of~\cref{algo:PLRD} with a cluster-robust analogue, and optimizing the minimax weights to account for within-cluster correlation in the objective \citep{noack2026inference}. 

\section*{Replication Files}
The \texttt{plrd} \texttt{R} package implementing the proposed method, along with replication files, is publicly available at \url{https://github.com/ghoshadi/plrd}.

\section*{Acknowledgments}
We are grateful to Evan Munro for sharing results on the behavior of existing methods for RDDs under generative adversarial simulation following \citep{athey2021using}. We are also grateful to Tim Armstrong, P.~Aronow, Max Farrell, Paul Goldsmith-Pinkham, Calumn Hamilton, Michal Koles\'ar, Christoph Rothe and Roc\'io Titiunik for insightful conversations and for comments on a previous version of this paper. This work was supported by the National Science Foundation under grant SES-2242876, the Office of Naval Research under grants N00014-17-1-2131 and N00014-19-1-2468, and a gift from Amazon.

\bibliographystyle{chicago}
\bibliography{references}

\clearpage
\appendix
\setcounter{section}{0}
\linespread{1.1}

\section{Proofs of the Main Results}
\subsection{Proof of \texorpdfstring{\cref{thm:validity}}{Theorem 1}}
Define $\wh{s}^2(\wh\g)=\sum_{i=1}^n \wh\g_i^2\wh\sigma_i^2$ as in \cref{algo:PLRD}. It follows from \citep{white1980heteroskedasticity} that $\wh{s}(\wh\g)/s(\wh\g)\Pto 1$; see \cref{lemma:ehw} for details. Next, fix any $\delta\in(0,1)$. In view of \eqref{eq:bias_ok} of the main paper and Slustky's lemma, $\P(A_n)\to 1$ where $$A_n:=\{|b(\wh\g)|\le \wh{b}+\delta\wh{s}(\wh\g)\}.$$ We show in \cref{lemma:Gaussian-CDF} that for any $s$, $h$, $t$ and $\delta>0$,
$$\inf_{|b|\,\le\, t+\delta s}\P(|b+s Z|\le h) \ge \inf_{|b|\,\le\, t}\P(|b+s Z|\le h) -\delta,$$
where $Z\sim \normal(0,1)$.
Using this with $s=\wh{s}(\wh\g)$, $h=\wh{h}_\alpha$ and $t=\wh{b}$, we get, for every $n$,
\begin{equation}\label{1-alpha-eps}
\begin{split}
    \inds{A_n}\P\left(|b(\wh\g)+\wh{s}(\wh\g) Z|\le \wh{h}_\alpha\,\Big|\, (X_i, Y_i)_{i=1}^n\right)&\ge \inds{A_n}\inf_{|b|\,\le \,\wh{b}\,+\,\delta\wh{s}(\wh{\gamma})}\P\left(|b+\wh{s}(\wh{\gamma})Z|\le \wh{h}_\alpha\right)\\
    &\ge \inds{A_n} \left(\inf_{|b|\,\le\, \wh{b}}\P(|b+\wh{s}(\wh\g) Z|\le \wh{h}_\alpha)-\delta\right)\\ 
    &\ge \inds{A_n} (1-\alpha-\delta),
\end{split}
\end{equation}
where the last inequality follows from the definition of $\wh{h}_\alpha$.
On the other hand, \cref{thm:main-CLT} and Slutsky's lemma imply that
$$Z_n:=\frac{\wh\tau_\texttt{PLRD}-\t-b(\wh\g)}{\wh{s}(\wh\g)}\dto Z\sim \normal(0,1).$$ Consequently,
\begin{align*}
    \liminf_{n\to\infty}\P(\t\in\mathcal{I}_\alpha) &= \liminf_{n\to\infty}\P(|\wh\tau_\texttt{PLRD}-\t|\le \wh{h}_\alpha) \\
    &=\liminf_{n\to\infty}
\P\left(|b(\wh\g)+\wh{s}(\wh\g) Z|\le \wh{h}_\alpha\right) \\
    &\geq\liminf_{n\to\infty}
\P(\{|b(\wh\g)+\wh{s}(\wh\g) Z|\le \wh{h}_\alpha\}\cap A_n) \\
&\ge 1-\alpha-\delta.\tag{using \eqref{1-alpha-eps} and $\P(A_n)\to 1$}
\end{align*}
Since $\delta\in(0,1)$ is arbitrary, we are through.
\qed
    
\subsection{Proof of \texorpdfstring{\cref{thm:main-CLT}}{main CLT}}\label{sec:asymptotic}

The optimization problem in \cref{algo:PLRD} is inherently related to its convex dual (see \citep{BoydVandenberghe} for a definition). 
We show in \cref{sec:implementation} that the dual problem is given by
\begin{equation}\label{dual:empirical-duplicate}
\begin{split}
    &\wh{\theta}_n(\sigma^2,B)=\argmin_{\theta=({\rho},\lambda)\in \Theta} \wh{L}_n(\theta;\sigma^2,B),\\
    &\wh{L}_n(\theta;\sigma^2,B):=\frac{1}{4n\sigma^2}\sum_{i=1}^n G^2(\theta; X_i-c, W_i)+\frac{\lambda_1^2}{4nB^2}+\frac{\lambda_2-\lambda_3}{n},
\end{split}
\end{equation}
where $\theta=({\rho},\lambda)$ varies in the set
\begin{equation}\label{def:dual-space}
    \Theta:=\left\{(f,\lambda)\mid f:[-\ell,\ell]\to \R, \, f(0)=f'(0)=f''(0)=0, \, f'' \text{ is $\lambda_1$-Lipschitz}, \, \lambda\in [0,\infty)\times \R^5\right\},
\end{equation}
and 
\begin{equation*}
    G({\rho},\lambda; x, w):={\rho}(x)+\lambda_2 w+\lambda_3(1-w)+\lambda_4 wx +\lambda_5 (1-w)x+\lambda_6 x^2.
\end{equation*}
In \cref{algo:PLRD}, we use sample splitting to estimate  $B$ (cf.~\cref{prop:B} of the main paper) and the conditional variance $\sigma^2=\mathbb{E}[\var(Y_i\mid X_i)]$. However, for ease of exposition, we assume throughout this section that we have $\wh{B}\rightarrow_p B$ and $\wh\sigma^2\rightarrow_p\sigma^2$ independently of the data, and 
 $$\wh\theta_n := \wh\theta_n(\wh\sigma^2,\wh{B})=\argmin_{\theta\,\in\,\Theta} \wh{L}_n(\theta;\wh\sigma^2,\wh{B}).$$
We can recover the primal solution $(\wh\g_n,\wh{t}_n\,)$ from the dual solution $\wh\theta_n$ as:
$$\wh\g_{n,i} = -\frac{1}{2\wh\sigma_i^2} \, G(\wh\theta_n; \, X_i-c, W_i),\quad \wh{t}_n = \frac{1}{2\wh{B}^2}\wh{\lambda}_{n,1}.$$
Next, define the population-level dual problem as:
\begin{equation}\label{dual:population}
\begin{split}
&\theta_n^*(\sigma^2,B)=\argmin_{\theta\,\in\,\Theta}L_n(\theta;\sigma^2,B), \\ &L_n(\theta;\sigma^2,B):=\frac{1}{4\sigma^2}\E\left[G^2(\theta; X-c, W)\right]+\frac{\lambda_1^2}{4nB^2}+\frac{\lambda_2-\lambda_3}{n}.
\end{split}
\end{equation}
Our proof strategy hinges on establishing fast-enough convergence of
\smash{$\wh\theta_n$} to \smash{$\theta_n^*(\sigma^2,B)$} to enable us
to read-off large-sample properties of \smash{$\wh{\tau}_{\mathtt{PLRD}}$}
from the solution to the population dual \eqref{dual:population}.

Although the empirical minimization problem \eqref{dual:empirical-duplicate} looks similar to a moment-estimation problem, it has some technical complications, \textit{e.g.},
\begin{enumerate}[label=(\roman*)]
    \item the dual parameter $\theta=(\rho,\lambda)$ involves a function $\rho$ and a vector $\lambda$,
    \item the set $\Theta$ in \eqref{def:dual-space} is \emph{not} the Cartesian product of a function space and a vector space,
    \item the population dual problem in \eqref{dual:population} involves the sample size $n$.
\end{enumerate}
To tackle these issues, 
we borrow tools from empirical process theory (see, {\it e.g.,} \citep{VW} for a textbook treatment).
The following result, proven in \cref{proof-of-lemma}, characterizes the rate of convergence of
the empirical dual optimizer $\wh\theta_n$ to the population dual optimizer $\theta_n^*$;
and implies (\cref{coro:lindeberg}) that the weights underlying \smash{$\wh{\tau}_{\mathtt{PLRD}}$}
satisfy a Lindeberg-type condition.

\begin{lemma}\label{thm:dual-rate-of-conv}
Let $(\wh{\rho}_n, \wh\lambda_n)=\wh\theta_n(\wh{\sigma}^2,\wh{B})$ be the optimizer of the empirical dual problem \eqref{dual:empirical-duplicate}, and $({\rho}_n^*, \lambda_n^*)=\theta_n^*(\wh{\sigma}^2, B)$ be the optimizer of the population dual problem \eqref{dual:population}. Assume that the pilot bandwidth $\ell$ in \cref{algo:PLRD} satisfies \eqref{eq:ell_decay}, and $\wh{B}\Pto B$. Then, under the conditions of \cref{thm:validity},
$$\|\wh{\rho}_n-\rho_n^*\|^2_{L^2(P)}+\|\wh{\lambda}_{n,-1}-\lambda_{n,-1}^*\|_2^2 = \O\left(\frac{\ell^{\,6}}{n^{1/4+6/7}}+ \frac{\ell^{\,3}}{n^{1+3/7}}\right).
$$
\end{lemma}

\begin{corollary}
\label{coro:lindeberg}
Under the conditions of \cref{thm:validity}, the weights $\wh{\gamma}_i$ underlying
$\wh{\tau}_{\mathtt{PLRD}}$ satisfy 
\begin{equation}
    \frac{\|\wh\g\|_\infty}{\|\wh\g\|_2}=\frac{\max_{1\le i\le n} |\wh\g_{i}|}{(\sum_{i=1}^n \wh\g_{i}^2)^{1/2}}\Pto 0.
\end{equation}
\end{corollary}

The proof of \cref{coro:lindeberg} is delegated to \cref{proof-of-corollary}.
\cref{coro:lindeberg} immediately implies that the central limit theorem claimed in
\cref{thm:main-CLT} holds; see the proof of \citep[Theorem 1]{optrdd2019} for details. 

We delegate showing that $s^2(\wh\g)=\O(n^{-6/7})$ (\cref{lemma:minimax-bias-rate}) and
$\wh{b}/{s}(\wh\g)=\O(1)$ (\cref{lemma:bias-ok}) to \cref{app:aux-results}. Finally, to verify \eqref{eq:bias_ok} of the main paper, we note that \cref{assu:part_lin,assu:smooth} yield
$b(\wh{\g}) = \sum_{i=1}^n \wh\g_i \mu_{W_i}(X_i) = B\sum_{i=1}^n\wh\g_i \rho(X_i-c)$, where
$$\rho(x-c):=(\mu_0(x)-\mu_0(c)-\mu'_0(c)(x-c)-\mu''_0(c)(x-c)^2/2)/B.$$
It follows from \cref{algo:PLRD} (with cross-fitted $\wh{B}$) that $|b(\wh\g)|\le B\wh{t}$ where $\wh{t}=\wh{b}/\wh{B}$.  Consequently,
\begin{equation}
\label{eq:bias_cmp}
\limsup_{n\to\infty} \frac{\left| b(\wh{\g})\right| - \wh{b}}{{s}(\wh\g)} - \frac{(B/\wh{B}-1)\wh{b}}{{s}(\wh\g)} \leq_p 0. 
\end{equation}
Furthermore, since $\wh{B}\Pto B$ and $\wh{b}/{s}(\wh\g)=\O(1)$, we have
$$\lim_{n\to\infty} \frac{(B/\wh{B}-1)\wh{b}}{{s}(\wh\g)} =_p 0, $$
and so \eqref{eq:bias_cmp} implies \eqref{eq:bias_ok} of the main paper.\qed


\subsection{Proof of \texorpdfstring{\cref{thm:dual-rate-of-conv}}{Lemma 1}}\label{proof-of-lemma}

\begin{proof}
For simplicity in notation, we assume without loss of generality that the regression discontinuity threshold is $c=0$. We split the proof into some auxiliary results, namely, \cref{lemma:emp-dual-close-to-population-dual,lemma:strongly-cvx,lemma:uniform-concentration,lemma:one-to-minus-one,lemma:minimax-bias-rate} in \cref{app:aux-results}. 
First, 
using tools from empirical process theory (see \citep{VW} for a book-level treatment), we 
prove a uniform concentration bound for the loss functions of the  empirical and population dual problems. Recall that the dual parameter space is defined as $$\Theta(\ell):=\left\{(f,\lambda)\mid f:[-\ell,\ell]\to\R,\, f(0)=f'(0)=f''(0)=0, f'' \text{ is $\lambda_1$-Lipschitz}, \lambda\in [0,\infty)\times \R^5\right\}.$$ Define a neighborhood $\Theta_n(\ell,\delta,\eps)$ of the dual optimizer $\theta_n^*$, as follows:  $$\Theta_n(\ell,\delta,\eps):=\{(\rho,\lambda)\mid (\rho,\lambda)\in \Theta(\ell),\ |\lambda_1-\lambda_{n,1}^*|\le \delta \lambda_{n,1}^* \text{ and } \|\lambda_{-1}-\lambda^*_{n,-1}\|_2\le \eps \|\lambda_{n,-1}^*\|_2\},$$ where  $\ell,\delta,\eps\in(0,1)$ are fixed. We show in \cref{lemma:uniform-concentration} the following maximal inequality.
\begin{equation*}
\E \sup_{\theta\,\in\,\Theta_n(\ell,\delta,\eps)} |\wh{L}_n(\theta;\wh\sigma^2,\wh{B})-L_n(\theta;\wh\sigma^2, \wh{B})|\lesssim (\lambda_{n,1}^*)^2\frac{\ell^6}{\sqrt{n}}+\|\lambda_{n,-1}^*\|_2^2\frac{\sqrt{\eps}}{\sqrt{n}}+\lambda_{n,1}^*\|\lambda_{n,-1}^*\|_2\frac{\ell^3}{\sqrt{n}}.
\end{equation*}
Next, \cref{lemma:one-to-minus-one} simplifies the above bound to the following.
\begin{equation}\label{recall-maximal-inequality2}
\E \sup_{\theta\,\in\,\Theta_n(\ell,\delta,\eps)} |\wh{L}_n(\theta;\wh\sigma^2,\wh{B})-L_n(\theta;\wh\sigma^2, \wh{B})|\lesssim (\lambda_{n,1}^*)^2\frac{\ell^6}{\sqrt{n}}+\lambda_{n,1}^*\frac{\ell^3}{n^{3/2}}+\frac{\sqrt{\eps}}{n^{5/2}}.
\end{equation}
On the other hand, \cref{lemma:one-to-minus-one}  tells us that
\begin{equation*}
\sup_{\theta\,\in\,\Theta_n(\ell,\delta,\eps)} |{L}_n(\theta;\wh{\sigma}^2,\wh{B})-L_n(\theta;\wh{\sigma}^2, B)|\lesssim 
\frac{(\lambda_{n,1}^*)^2}{n}\left|\frac{1}{\wh{B}}-\frac{1}{B}\right|.
\end{equation*}
This, in conjunction with 
\eqref{recall-maximal-inequality2} tells us that for every $\eta\in(0,1)$, there exists $C=C(\eta)>0$ such that, for all sufficiently large $n$,
\begin{multline}\label{recall-maximal-ineq3}
\P\left(\sup_{\theta\,\in\,\Theta_n(\ell,\delta,\eps)} |\wh{L}_n(\theta;\wh\sigma^2,\wh{B})-L_n(\theta;\sigma^2,B)|\le C \left((\lambda_{n,1}^*)^2\frac{\ell^6}{\sqrt{n}}+\lambda_{n,1}^*\frac{\ell^3}{n^{3/2}}+\frac{\sqrt{\eps}}{n^{5/2}}+ \frac{(\lambda_{n,1}^*)^2}{n}\right)\right)\\
\ge 1-\eta.
\end{multline}
We also show in \cref{lemma:strongly-cvx} that 
\begin{equation}\label{recall-strong-convexity2}
     L_n(\theta;\sigma^2,B)-L_n(\theta_n^*;\sigma^2,B)
    \ge \frac{\alpha}{4\sigma^2}\left(\|\rho-\rho_n^*\|^2_{L^2(P)}+\|\lambda_{-1}-\lambda_{n,-1}^*\|_2^2\right)+\frac{(\lambda_1-\lambda_{n,1}^*)^2}{4nB^2}-R_n,
\end{equation}
for some $\alpha\in (0,1)$, where $0\le R_n\le \ell^3(\lambda_1+\lambda_{n,1}^*)\|\lambda_{-1}-\lambda_{n,-1}^*\|_2/2\sigma^2$. 
Note, for $(\rho,\lambda)\in \Theta_n(\ell,\delta,\eps)$, 
\begin{equation}\label{bound-on-remainder2}
    R_n\le\frac{1}{2\sigma^2} \ell^3 (2+\delta)\lambda_{n,1}^*\eps\|\lambda_{n,-1}^*\|_2\le \frac{3}{2\sigma^2} \lambda_{n,1}^*\|\lambda_{n,-1}^*\|_2 \ell^3\eps.
\end{equation}
Next, we show in \cref{lemma:emp-dual-close-to-population-dual} that for every $\delta,\eta\in(0,1)$, and $\eps=O(n^{-1/4})$, it holds that
$$\P\left(\wh{\theta}_n \in \Theta_n(\ell,\delta,\eps)\right)\ge  1-\eta,$$
for all sufficiently large $n$.
The following inequalities hold whenever the event in \eqref{recall-maximal-ineq3} holds as well as $\wh{\theta}_n\in \Theta_n(\ell,\delta,\eps)$ with $\eps=O(n^{-1/4})$, i.e., the following inequalities hold with probability at least $1-2\eta$ for all sufficiently large $n$.
\begin{align*}
      &\frac{\alpha}{4\sigma^2}\left(\|\wh{\rho}_n-\rho_n^*\|^2_{L^2(P)}+\|\wh{\lambda}_{n,-1}-\lambda_{n,-1}^*\|_2^2\right)+\frac{(\wh{\lambda}_{n,1}-\lambda_{n,1}^*)^2}{4nB^2}\\
      &\le \frac{3}{2\sigma^2} \lambda_{n,1}^*\|\lambda_{n,-1}^*\|_2 \ell^3\eps + L_n(\wh{\theta};\sigma^2,B)-L_n(\theta_n^*;\sigma^2,B)\tag{using \eqref{recall-strong-convexity2} and \eqref{bound-on-remainder2}}\\
      &\le \frac{3}{2\sigma^2} \lambda_{n,1}^*\|\lambda_{n,-1}^*\|_2 \ell^3\eps + \underbrace{\wh{L}_n(\wh{\theta};\wh\sigma^2,\wh{B})-\wh{L}_n(\theta_n^*;\wh\sigma^2,\wh{B})}_{\le 0}+ 2\sup_{\theta\in \Theta_n(\ell,\delta,\eps)} \left|\wh{L}_n(\theta;\wh\sigma^2,\wh{B})-L_n(\theta;\sigma^2,B)\right|\tag{using triangle inequality and the definition of $\wh{\theta}_n$}\\
      &\lesssim  \lambda_{n,1}^*\|\lambda_{n,-1}^*\|_2 \frac{\ell^3}{n^{1/4}} + (\lambda_{n,1}^*)^2\frac{\ell^6}{\sqrt{n}}+\lambda_{n,1}^*\frac{\ell^3}{n^{3/2}}+\frac{1}{n^{5/2+1/8}}+ \frac{(\lambda_{n,1}^*)^2}{n}.\tag{using \eqref{recall-maximal-ineq3}}
\end{align*}
Finally, \cref{lemma:one-to-minus-one} tells us that $\|\lambda_{n,-1}^*\|_2 =O(\lambda_{n,1}^*\ell^3+n^{-1})$, while \cref{lemma:minimax-bias-rate} tells us that $\lambda_{n,1}^*=\O(n^{-3/7})$. This simplifies the above upper bound to the following.
\begin{align*}
      &\frac{\alpha}{4\sigma^2}\left(\|\wh{\rho}_n-\rho_n^*\|^2_{L^2(P)}+\|\wh{\lambda}_{n,-1}-\lambda_{n,-1}^*\|_2^2\right)+\frac{(\wh{\lambda}_{n,1}-\lambda_{n,1}^*)^2}{4nB^2}
      \\&\qquad\qquad\qquad
      \lesssim_\eta (\lambda_{n,1}^*)^2 \frac{\ell^6}{n^{1/4}} +\lambda_{n,1}^*\frac{\ell^3}{n^{3/2}}+ \frac{(\lambda_{n,1}^*)^2}{n}\\
      &\qquad\qquad\qquad\lesssim_\eta \frac{\ell^6}{n^{1/4+6/7}} +\frac{\ell^3}{n^{3/2+3/7}}+ \frac{1}{n^{1+6/7}},
      \end{align*} with probability at least $1-2\eta$, for all sufficiently large $n$. Finally, we can ignore the last term in the above display, since $\ell$ satisfies \eqref{eq:ell_decay}.
 This completes the proof. \end{proof}

\subsection{Proof of \texorpdfstring{\cref{coro:lindeberg}}{Corollary 2}}\label{proof-of-corollary}

\begin{proof}
For the PLRD weights $\wh\g_i$ derived in \cref{algo:PLRD}, we now show that 
\begin{equation*}
    \frac{\|\wh\g\|_\infty}{\|\wh\g\|_2}=\frac{\max_{1\le i\le n} |\wh\g_{i}|}{(\sum_{i=1}^n \wh\g_{i}^2)^{1/2}}\Pto 0.
\end{equation*}
Assume without loss of generality that the threshold for treatment is at $c=0$.
Recall that we use sample splitting in \cref{algo:PLRD} and write $\wh\g_i = \wh\g_{(k),i}/2$ if $i\in I_k$ ($k=1,2$). Thus, 
$$\frac{\|\wh\g\|_\infty}{\|\wh\g\|_2}=\frac{\max_{1\le i\le n} |\wh\g_{i}|}{(\sum_{i=1}^n \wh\g_{i}^2)^{1/2}}\lesssim \frac{\|\wh\g_{(1)}\|_\infty}{\|\wh\g_{(1)}\|_2}+\frac{\|\wh\g_{(2)}\|_\infty}{\|\wh\g_{(2)}\|_2}.$$
Since $\wh\g_{(1)}$ and $\wh\g_{(2)}$ are identically distributed, it suffices to show the following.
\begin{equation}\label{eqn:asymp-ignore-condition}
    \frac{\|\wh\g_{(1)}\|_\infty}{\|\wh\g_{(1)}\|_2}=\frac{\max_{i\in I_1} |\wh\g_{(1),i}|}{(\sum_{i\in I_1} \wh\g_{(1),i}^2)^{1/2}}\Pto 0.
\end{equation}
Towards showing \eqref{eqn:asymp-ignore-condition}, note that by Cauchy-Schwarz inequality,
$$|I_1|\sum_{i\in I_1} \wh\g_{(1),i}^2=\sum_{i\in I_1} (2W_i-1)^2 \sum_{i\in I_1} \wh\g_{(1),i}^2\ge \bigg(\sum_{i\in I_1} (2W_i-1)\wh\g_{(1),i}\bigg)^2=4.$$
This yields the bound $\|\wh\g_{(1)}\|_2\gtrsim n^{-1/2}$. For the numerator in \eqref{eqn:asymp-ignore-condition} use triangle inequality to write
\begin{align*}
    \sqrt{n}\max_{i\in I_1} |\wh\g_{(1),i}| &\le  \sqrt{n}\max_{i\in I_1} |\g_{(1),i}^*|+ \sqrt{n}\max_{i\in I_1} |\wh\g_{(1),i}-\g_{(1),i}^*|,
\end{align*}
where
$\g_{(1),i}^* := -(\rho_n^*(X_i)+ \lambda_{n,-1}^* v(X_i, W_i))/2\wh\sigma_{(2)}^2$,
with $v(x,w)=(w,1-w, xw, x(1-w), x^2)^\top$, and $$(\rho_n^*, \lambda_{n}^*) := \argmin_{(\rho,\lambda)\in \Theta} \frac{1}{4\sigma^2_{(2)}}\E\left[G^2(\rho, \lambda; X, W)\right]+\frac{\lambda_1^2}{4nB^2_{(2)}}+\frac{\lambda_2-\lambda_3}{n},$$
with $\Theta$ as defined in \eqref{dual-param-space-duplicate}. 
For any $(\rho,\lambda)\in \Theta$, $\rho^{(j)}(0)=0$ for $j=0,1,2$, and $\rho''$ is $\lambda_{n,1}^*$-Lipschitz, so a Taylor expansion of ${\rho}(x)$ around $x=0$ gives us
\begin{equation}\label{taylor-expand-rho-Theta}
    |{\rho}(x)|=|\rho(x)-\rho(0)-\rho'(0)x|=\frac{|\rho''(\xi)-\rho''(0)|}{|\xi - x|}\left|\xi-x\right|\frac{x^2}{2}\le \lambda_{n,1}^*\frac{\ell^3}{2}.
\end{equation}
Consequently,
$$\max_{i\in I_1} |\g_{(1),i}^*|\lesssim \sup_x \left|\rho_n^*(x)+v(x,1(x>0))^\top\lambda_{n,-1}^*\right|\lesssim \lambda_{n,1}^* \ell^3+\|\lambda_{n,-1}^*\|_2,$$
where we used \eqref{taylor-expand-rho-Theta} to bound $\sup_x|\rho_n^*(x)|$ and Cauchy-Schwarz inequality to bound the other term. Since $B_{(2)}$ and $\sigma_{(2)}^2$ are independent of the data in the fold $I_1$, \cref{lemma:one-to-minus-one} applies here and tells us that $\|\lambda_{n,-1}^*\|_2=O(\lambda_{n,1}^*\ell^3+n^{-1})$. Consequently, we have
$$\sqrt{n}\max_{i\in I_1} |\g_{(1),i}^*|\lesssim \sqrt{n}\left(\lambda_{n,1}^*\ell^3+n^{-1}\right).$$
On the other hand, $
   \max_{i\in I_1} |\wh\g_{(1),i}-\g_{(1),i}^*|\lesssim \|\wh\rho_n-\rho_n^*\|_\infty + \|\wh\lambda_{n,-1}-\lambda_{n,-1}^*\|_2$. It follows from \eqref{Taylor-expand-rho} that $\|\wh\rho_n-\rho_n^*\|_\infty\le (\wh\lambda_{n,1}+\lambda_{n,1}^*)\ell^3$. Applying
\cref{thm:dual-rate-of-conv}, we deduce that
$$\|\wh\lambda_{n,-1}-\lambda_{n,-1}^*\|_2
=\O\left( \lambda_{n,1}^* \frac{\ell^3}{n^{1/8}} +\sqrt{\lambda_{n,1}^*}\frac{\ell^{3/2}}{\sqrt{n}}+ \frac{\lambda_{n,1}^*}{n^{3/4}}\right).$$ 
Putting this all together,
$$
  \sqrt{n} \max_{i\in I_1} |\wh\g_{(1),i}-\g_{(1),i}^*|\lesssim (\wh\lambda_{n,1}+\lambda_{n,1}^*)\ell^3\sqrt{n}+\sqrt{\lambda_{n,1}^*}\ell^{3/2}+\lambda_{n,1}^*n^{-1/4}.
$$
Further, we show in \cref{lemma:minimax-bias-rate} that $\wh\lambda_{n,1}$ and $\lambda_{n,1}^*$ are both $O(n^{-3/7})$. 
Consequently,
\begin{align*}
    \sqrt{n}\max_{i\in I_1} |\wh\g_{(1),i}| &\le  \sqrt{n}\max_{i\in I_1} |\g_{(1),i}^*|+ \sqrt{n}\max_{i\in I_1} |\wh\g_{(1),i}-\g_{(1),i}^*|=\O\left(n^{-3/7} \ell^3\sqrt{n}\right)=\O\left(n^{1/14}\ell^3\right),
\end{align*}which is $\o(1)$ since $\ell=o(n^{-1/12})$. \end{proof}

\section{Safeguard against Curvature Changes}
\label{sec:cuv_change}

While \cref{assu:part_lin} provides a parsimonious and empirically realistic model for treatment effect heterogeneity in many applications, it can in principle be violated. To safeguard against such violations, our software incorporates an automatic diagnostic that tests for difference in curvatures. 
The diagnostic works as follows. Before running the main PLRD procedure, we conduct an ANOVA test comparing a restricted cubic polynomial model (where treatment interacts only linearly with $X_i - c$, consistent with \cref{assu:part_lin}) against an unrestricted model (where treatment interacts with all polynomial terms in $X_i - c$). If the test rejects at significance level $\alpha'$ (default 0.1\%), the method proceeds with a more flexible specification that only assumes \cref{assu:smooth} but not \cref{assu:part_lin}. Otherwise, it uses the default specification under both assumptions. This adaptive approach helps explain the strong empirical performance of \texttt{plrd} across all simulation settings in \cref{tab:simulations}, including those where \cref{assu:part_lin} fails.

When the flexible specification is selected, estimation of the Lipschitz constant $B$ changes to accommodate potentially different curvature on either side of the threshold. Rather than fitting a single cubic polynomial as in \eqref{eq:poly}, we fit \emph{separate} cubic regressions for treated and control units:

\begin{proposition}[Consistent estimation of the Lipschitz constant $B$ under different curvatures]
\label{prop:Bdiffcurv}
Suppose that \cref{assu:smooth} holds and that $X_i$ has a density bounded away
from 0 in a neighborhood of $c$, and $B := \max\{|\mu_0'''(c)|,|\mu_1'''(c)|,\eps\}$. Estimate $\wh B$ via local cubic regression: For $w\in\{0,1\}$ run
\begin{equation}
\label{eq:polydiffcurv}
\wh\beta_w:=\argmin_{\beta_w\in \R^4} \sum_{\substack{W_i = w \\ |X_i - c| \leq \ell}} \left(Y_i - \beta_{w,0}-\beta_{w,1} (X_i-c)-\frac{\beta_{w,2}}{2}(X_i-c)^2-\frac{\beta_{w,3}}{6}(X_i-c)^3\right)^2,
\end{equation}
with bandwidth $\ell$ scaling as in \eqref{eq:ell_decay}, and set
$\wh B = \max\{|\beta_{0,3}|, |\beta_{1,3}|,\eps\}$.
Then, $\wh B\rightarrow_p B$ as $n\to\infty$.
\end{proposition}
With this estimated $\wh B$ in hand, the subsequent minimax optimization in \cref{algo:PLRD} adjusts the constraint set to allow for different curvature functions $\rho_0$ and $\rho_1$ on either side of the cutoff (see \eqref{eq:point_est_diffcurv} versus \eqref{eq:point_est}). The optimization problem remains computationally tractable via quadratic programming, and the remainder of the PLRD procedure—including cross-fitting, variance estimation, and confidence interval construction—proceeds similarly to the default case. We note that a formal proof of asymptotic validity under this flexible specification would follow similar arguments to those in \cref{thm:validity}, with nearly identical technical steps that are omitted here.

\section{Details on Baseline Methods}\label{app:ablation-analysis}

In this section, we give further details on the baseline methods \texttt{rdrobust} and
\texttt{RDHonest} considered in our experiments, and provide further discussion of how
our modeling assumptions differ from those used to-date in software for bias-aware inference
in RDDs.

\subsection{Bias-Corrected Confidence Intervals}\label{sec:rdrobust}

Define the local polynomial estimator of order $p$ as
$
\wh{\tau}_p(h_n):=\wh{\mu}_{1,\, p}(h_n)-\wh{\mu}_{0,\, p}(h_n),
$
where $\widehat{\mu}_{1,\, p}(h_n)$ and $\widehat{\mu}_{0,\, p}(h_n)$ are the intercepts from local polynomial regressions of order $p$ fitted separately on each side of the threshold. More precisely, for $w\in\{0,1\}$, compute
\begin{align*}
\wh{\mu}_{w,\, p}(h_n)&:=\mathbf{e}_0^{\prime} \wh{\boldsymbol{\beta}}_{w,\, p}(h_n),\\ \wh{\boldsymbol{\beta}}_{w,\, p}(h_n)&:=\argmin _{\boldsymbol{\beta} \in \mathbb{R}^{p+1}} \sum_{i\,:\,W_i=w} K(|X_i-c|/h_n)\left(Y_i-\mathbf{r}_p(X_i-c)^{\top} \boldsymbol{\beta}\right)^2,
\end{align*}
where $\mathbf{e}_0=(1,0, \ldots, 0) \in \mathbb{R}^{p+1}$ is the first unit vector, $\mathbf{r}_p(x)=\left(1, x, \ldots, x^p\right)^{\top}$, $K(\cdot)$ is a kernel, and $h_n$ is a bandwidth. The first-order bias of $\wh\tau_p(h_n)$ is \citep{CCT2014robustCI,imbens2012optimal}
$$B_p(h_n)= \frac{h^{p+1}}{(p+1)!}\left[\mu_{1}^{(p+1)}(c+)-\mu_{0}^{(p+1)}(c-)\right], $$
where $\mu_w^{(p+1)}$ denotes the derivative of $\mu_w(\cdot)$ of order $p+1$. \cite{CCT2014robustCI} estimate the above bias using a higher-order local polynomial of order $q>p$ with a separate pilot bandwidth $b_n$, thus compute the bias-corrected estimator as $$\wh\tau^\mathrm{bc}_{p,q}(h_n,b_n) := \wh\tau_{p}(h_n)-h_n^{p+1}\wh{B}_{n,p,q}(h_n,b_n),$$ where $\wh{B}_{n,p,q}(h_n,b_n)$ is the estimated bias term $B_p(h_n)$. The confidence interval they construct based on the bias-corrected estimator takes the following form:
\begin{equation}\label{ci-bc}
    \wh{\text{CI}}_{1-\alpha}^{\text{bc}}(h_n, b_n) = \left( \wh\tau^\mathrm{bc}_{p,q}(h_n,b_n) \pm z_{1-\alpha/2}  \wh{\sigma}(h_n) \right),
\end{equation}
where $z_{1-\alpha/2}$ is the critical value from the standard normal distribution, and $\wh{\sigma}(h_n)$ is the standard error (for which \citep{CCT2014robustCI} provide two options: plug-in and fixed matches). The above CI construction is implemented in the \texttt{rdrobust} package, with  $p=1$ and $q=2$ as default.

 The CI constructed in \eqref{ci-bc} might not perform well in finite samples because the estimator of the bias introduces additional variability. To address this, \citep{CCT2014robustCI} also propose an alternative, more robust CI construction using the bias-corrected estimator $\wh\tau^\mathrm{bc}_{p,q}(h_n,b_n)$ as a starting point. They construct this robust CI by inflating the CI in \eqref{ci-bc} with a more conservative estimator of the standard error:
 \begin{equation}\label{ci-rbc}
    \wh{\text{CI}}_{1-\alpha}^{\text{rbc}}(h_n, b_n) = \left( \wh\tau^\mathrm{bc}_{p,q}(h_n,b_n) \pm z_{1-\alpha/2}  \sqrt{\wh{V}^\text{bc}_{n,p,q}} \right),
\end{equation}
where $\wh{V}^\text{bc}_{n,p,q}$ is an estimator of the asymptotic variance that accounts for both the variability in estimation and the bias-correction. In \cref{tab:wgan,tab:real-data-CIs,tab:simulations}, the columns marked as \ag{\texttt{rdrobut} correspond to the confidence intervals computed using \eqref{ci-rbc}}.

\subsection{Bias-Aware Confidence Intervals}\label{sec:bias-aware-recap}

The bias-aware inference approach \citep[e.g.,][]{ArmstrongKolesar2018,ArmstrongKolesar2020,optrdd2019} provides a different approach for obtaining confidence intervals for $\t$. Note that the LLR estimator takes the form $\wh\tau_\texttt{LLR}=\sum_{i=1}^n \g_{i,\texttt{LLR}} Y_i$ for some weights $\g_{i,\texttt{LLR}}$ that are functions of the running variable. In fact, the same also holds more generally for the estimator $\wh{\tau}_p(h_n)$ defined in \cref{sec:rdrobust}. Next, one makes the observation that although $\wh\tau_\texttt{LLR}$ is motivated by
a regression problem, the rest of our inference procedure hinges upon on properties that hold for any linear estimators of the form $\wh\tau(\g)=\sum_{i=1}^n \g_i Y_i$. Given by this observation, we can systematically do better by directly finding weights $\g_i$ that minimize the mean squared error (MSE) of the estimator $\wh\tau(\g)$ uniformly over a class $\cal{M}$ of conditional response functions $\mu_{w}$, i.e., 
\begin{equation}\label{bias-aware-weights}
    \wh\g :=\argmin_\g \sup_{\mu_{w}\,\in\,\cal{M}} \text{MSE}\left(\wh\tau(\g)\mid \{X_i\}_{i=1}^n\right).
\end{equation}
The two most widely considered classes to date \citep{ArmstrongKolesar2018,KolesarRothe2018,optrdd2019} are
\begin{align}
&\cal{M}_{\text{T2},\, M} =\left\{\mu_{w}(\cdot):|\mu_{w}(x)-\mu_{w}(c)-\mu_{w}'(c)x|\le M(x-c)^2/2\text{ for every }x,\text{ and } w=0,1\right\}, \label{def:T2}\\
&\cal{M}_{\text{H2},\, M} =\left\{\mu_{w}(\cdot): \mu'_{w}(x) \text{ is $M$-Lipschitz for all $x$ and } w=0,1\right\},\label{def:H2}
\end{align}
where $M>0$ is a constant. The first option, considered in early work by \citep{ArmstrongKolesar2018},
justifies accuracy of a 2nd-order Taylor expansion at $c$ and allows for explicit, closed-form solutions
to the minimax linear estimation problem. The second option, which simply takes all functions to belong
to a 2nd-order H\"older class, was used by \citep{KolesarRothe2018} and \citep{optrdd2019} and also
adopted for the \texttt{RDHonest} package by \citep{ArmstrongKolesar2018}.

Fixing a value of $M$, we numerically solve the optimization problem \eqref{bias-aware-weights} and obtain an upper bound on the worst-case bias of the resulting estimator as a by-product. We then follow \citep{ImbensManski2004} to construct a bias-aware confidence interval by inflating the Gaussian confidence interval accounting for the worst-case bias of $\tau(\wh\g)$ under the function class $\cal{M}$. These bias-aware confidence intervals provide uniformly valid coverage in finite samples,
provided the function class $\cal{M}$ and the distribution of $Y_i$ given $X_i$ is Gaussian \citep{ArmstrongKolesar2018}.\footnote{If the conditional distribution of $Y_i$ is non-Gaussian, one can recover exact finite-sample results using the bias-aware approach by using finite-sample concentration results instead of Gaussian approximation.} In practice, however, smoothness constant $M$ is typically unknown; both \citep{ArmstrongKolesar2020} and \citep{optrdd2019} suggest using some rule-of-thumb estimates.

\subsection{The Choice of Function Class}

One of the key insights our paper provides is that the choice of the function class $\cal{M}$ in the minimax problem \eqref{eqn:minimax-problem0} plays a crucial role in determining the conservativeness (and thus the practical usefulness) of the resulting bias-aware method. 
The function classes used in the existing bias-aware literature ({\it e.g.}, the classes $\cal{M}_{\text{T2},\, M}$ and $\cal{M}_{\text{H2},\, M}$ defined in \eqref{def:T2} and \eqref{def:H2} respectively) are much modest in terms of assumed smoothness as compared to the the bias-corrected methods used in practice (see, {\it e.g.}, \citep[Theorem 1]{CCT2014robustCI}). Here, on the other hand, use the class of all functions $\mu_{w}(\cdot)$ that satisfy \cref{assu:part_lin} and for which $\mu_{0}''(\cdot)$ is $\wh{B}$-Lipschitz, where we obtain $\wh{B}$ via a global cubic regression and cross-fitting.  While this might appear as a minor modification from a theoretical point of view, it materially affects the width of our intervals in practice, as we demonstrate in \cref{tab:wgan,tab:real-data-CIs,tab:simulations}. 

To illustrate the difference between the conservativeness of our function class versus the function classes used by existing bias-aware approaches, we plot in \cref{fig:enter-label} the regression functions corresponding to the worst-case bias obtained by our method and $\cal{M}_{\text{H2},\, M}$ as used in \texttt{RDHonest}. This plot demonstrates that the classes $\cal{M}_{\text{H2},\, M}$ can often be unreasonably broad, and include regression functions that are \emph{ex ante} implausible or rarely occur in practice. This explains why the resulting bias-aware confidence intervals can be substantially more conservative than the bias-corrected ones (see, {\it e.g.},~\cref{tab:wgan,tab:real-data-CIs,tab:simulations}). The empirical success of our PLRD method that we demonstrate in \cref{sec:empirical-assessment} is, in fact, rooted in this choice of the function class over which one computes the worst-case MSE.

\begin{figure}[t]
\centering
\includegraphics[width=0.495\linewidth]
{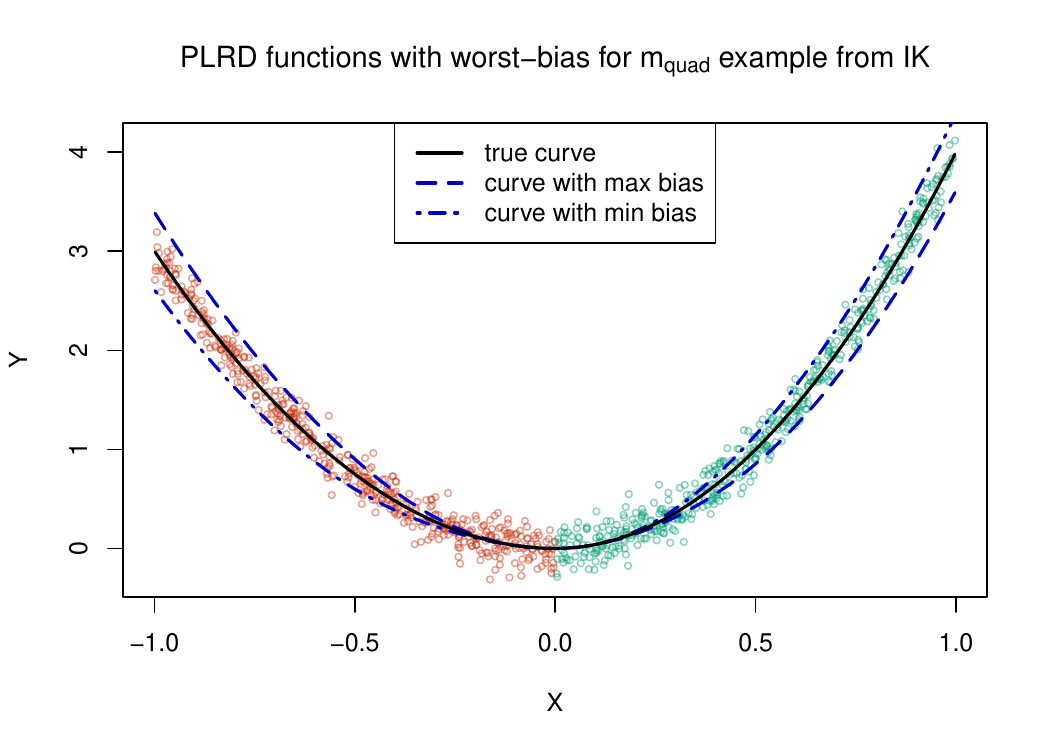}
\includegraphics[width=0.495\linewidth]
{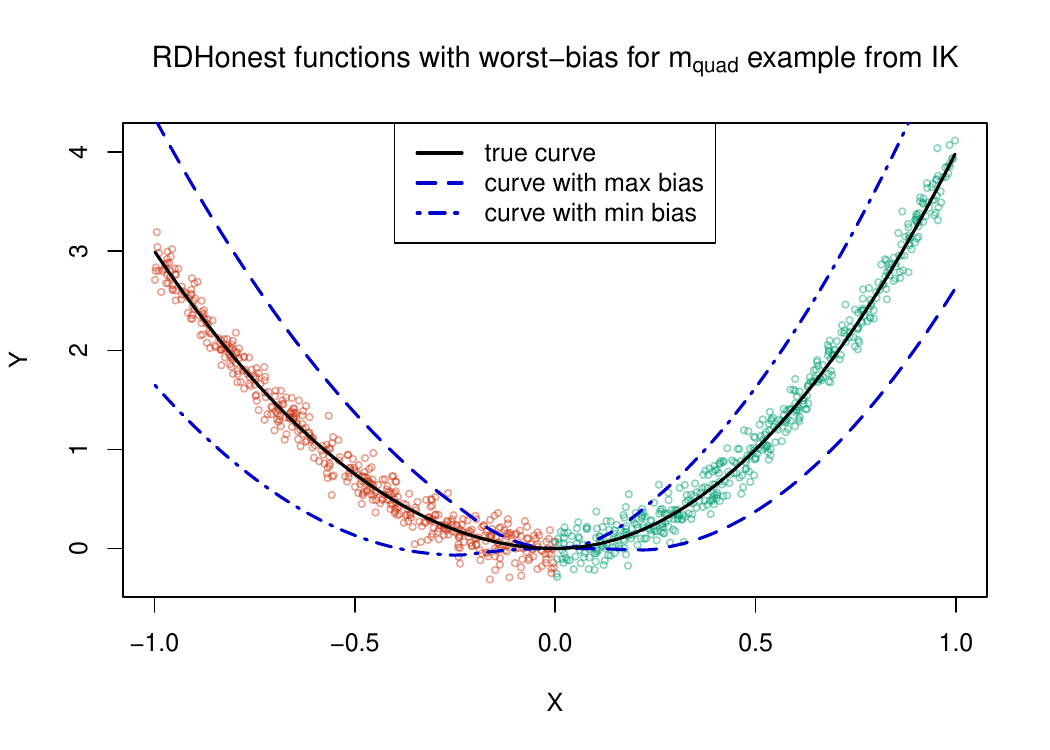}
\includegraphics[width=0.495\linewidth]
{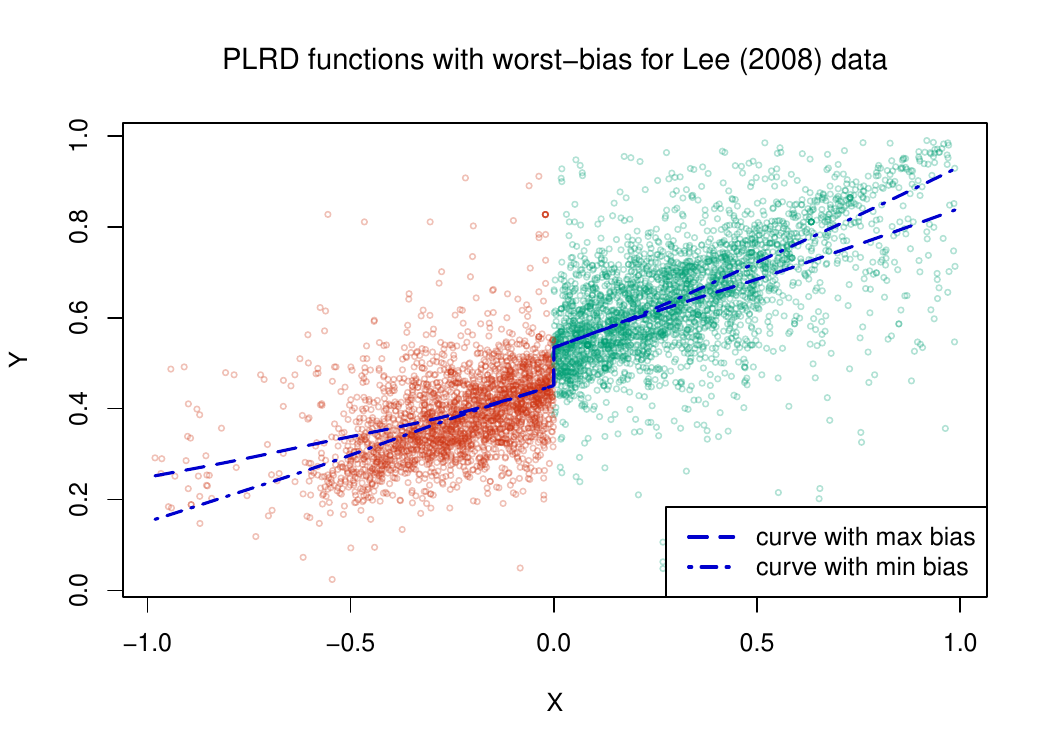}
\includegraphics[width=0.495\linewidth]
{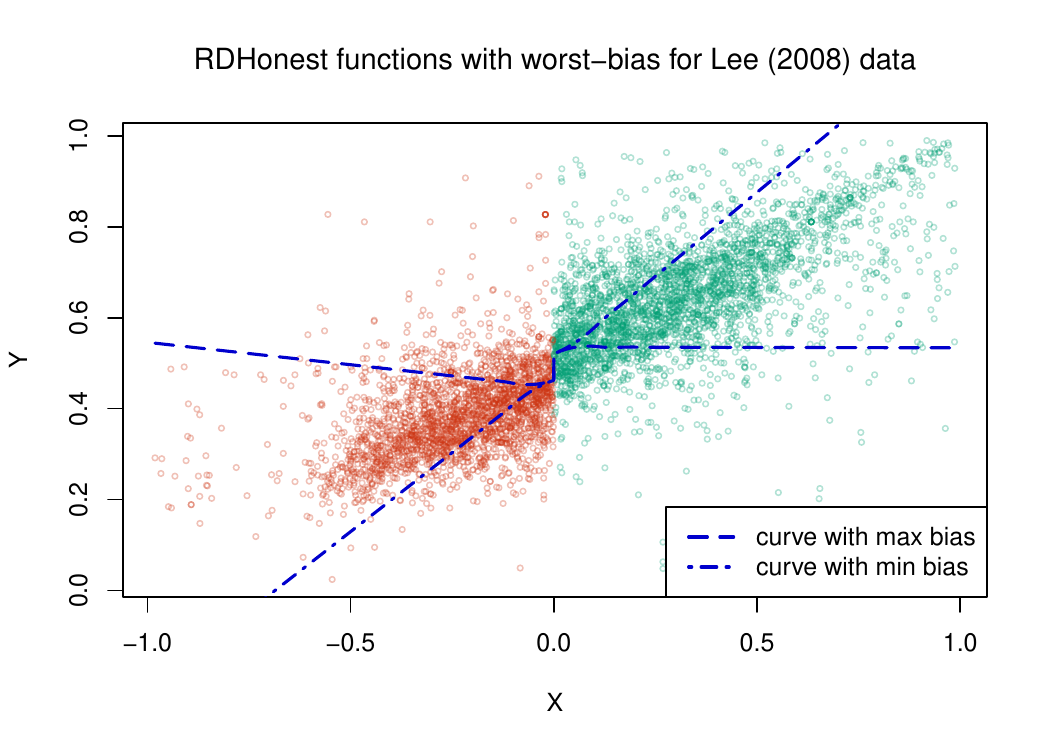}
    \caption{Functions $\mu_{w}(\cdot)$ with worst-case bias in the RDHonest and PLRD function classes. The top two plots use the $m_\text{quad}$ example from \citep{imbens2012optimal}, while the bottom two plots use the data from \citep{lee2008}. In each plot, black lines show $\mu_{w}(\cdot)$ functions with maximum (solid) or minimum (dotted) bias in the respective function class. This figure illustrates how the RDHonest function class may accommodate in-hindsight-unreasonable functions for several practical settings.}
    \label{fig:enter-label}
\end{figure}

\subsection{Ablation Analysis}

Here we present an empirical comparison of our PLRD method with other alternative methods. We summarize our empirical results in \cref{tab:wgan-ablation}, which complements the results in \cref{tab:wgan} of the main paper.
\cref{tab:wgan-ablation} compares our PLRD method with four other alternative approaches:

\begin{table}[t]
\centering
\scriptsize
\setlength{\tabcolsep}{2pt}
\renewcommand{\arraystretch}{1.1}
\resizebox{\linewidth}{!}{
\begin{tabular}{@{}l*{15}{c}@{}}
\toprule
& \multicolumn{3}{c}{\textbf{Alt.~(a)}} 
& \multicolumn{3}{c}{\textbf{Alt.~(b)}} 
& \multicolumn{3}{c}{\textbf{Alt.~(c)}} 
& \multicolumn{3}{c}{\textbf{Alt.~(d)}} 
& \multicolumn{3}{c}{\texttt{plrd}} \\
\cmidrule(lr){2-4}
\cmidrule(lr){5-7}
\cmidrule(lr){8-10}
\cmidrule(lr){11-13}
\cmidrule(lr){14-16}
\textbf{Dataset}
& \textbf{cvrg} & \textbf{width} & \textbf{rmse}
& \textbf{cvrg} & \textbf{width} & \textbf{rmse}
& \textbf{cvrg} & \textbf{width} & \textbf{rmse}
& \textbf{cvrg} & \textbf{width} & \textbf{rmse}
& \textbf{cvrg} & \textbf{width} & \textbf{rmse} \\
\midrule
Lee08
& \good{96.7\%} & 0.040 & 0.010
& \good{95.8\%} & 0.049 & 0.012
& \good{98.4\%} & 0.044 & 0.009
& \good{99.9\%} & 0.041 & \better{0.007}
& \good{98.4\%} & 0.044 & 0.009 \\

CFT15
& \good{96.6\%} & 5.513 & 1.316
& \good{96.0\%} & 6.791 & 1.661
& \good{98.3\%} & 5.643 & 1.189
& \good{99.9\%} & 9.731 & 1.535
& \good{96.2\%} & 3.865 & \better{0.937} \\

LM07
& \good{96.5\%} & 5.870 & 1.599
& \good{96.0\%} & 7.175 & 2.003
& \good{97.6\%} & 6.311 & 1.670
& \bad{87.9\%} & 3.526 & 1.182
& \good{96.5\%} & 3.746 & \better{1.144} \\

Mey14
& \good{96.1\%} & 7.021 & 1.732
& \good{95.5\%} & 8.696 & 2.173
& \good{97.0\%} & 6.503 & 1.498
& \good{97.7\%} & 7.590 & 1.644
& \good{96.7\%} & 4.682 & \better{1.159} \\

Mat08 (read.)
& \bad{87.2\%} & 0.056 & 0.018
& \good{95.0\%} & 0.072 & 0.018
& \good{96.7\%} & 0.059 & \better{0.014}
& \bad{94.6\%} & 0.085 & 0.022
& \good{96.7\%} & 0.059 & \better{0.014} \\

Mat08 (math)
& \good{94.7\%} & 0.049 & 0.013
& \good{94.9\%} & 0.059 & 0.015
& \good{97.6\%} & 0.051 & \better{0.011}
& \bad{92.1\%} & 0.074 & 0.022
& \good{97.6\%} & 0.051 & \better{0.011} \\

JL04 (read.)
& \good{96.1\%} & 0.077 & 0.020
& \good{94.9\%} & 0.101 & 0.027
& \good{94.9\%} & 0.083 & 0.021
& \good{97.3\%} & 0.072 & \better{0.016}
& \good{97.2\%} & 0.078 & 0.019 \\

JL04 (math)
& \good{96.5\%} & 0.065 & 0.016
& \good{95.1\%} & 0.084 & 0.022
& \good{95.0\%} & 0.069 & 0.018
& \good{96.9\%} & 0.070 & 0.016
& \good{97.6\%} & 0.060 & \better{0.013} \\

Ore06
& \good{96.3\%} & 0.212 & 0.055
& \good{96.5\%} & 0.400 & 0.123
& \good{98.5\%} & 0.245 & 0.052
& \good{98.3\%} & 0.180 & 0.038
& \good{95.9\%} & 0.124 & \better{0.031} \\

AMT22 (4th grade)
& \good{96.5\%} & 0.115 & 0.028
& \bad{94.4\%} & 0.117 & 0.032
& \good{96.8\%} & 0.096 & \better{0.022}
& \bad{62.8\%} & 0.086 & 0.041
& \good{96.8\%} & 0.096 & \better{0.022} \\

AMT22 (8th grade)
& \good{96.6\%} & 0.121 & 0.030
& \good{96.1\%} & 0.147 & 0.037
& \good{96.8\%} & 0.108 & 0.026
& \good{96.5\%} & 0.079 & 0.019
& \good{94.9\%} & 0.071 & \better{0.018} \\

AMT22 (HM)
& \good{97.0\%} & 0.066 & 0.016
& \good{96.0\%} & 0.081 & 0.020
& \good{97.5\%} & 0.060 & 0.014
& \good{97.1\%} & 0.045 & \better{0.010}
& \good{95.9\%} & 0.040 & \better{0.010} \\
\midrule
Pure noise
& \good{96.2\%} & 1.611 & 0.394
& \good{94.9\%} & 2.053 & 0.525
& \good{97.1\%} & 1.678 & 0.374
& \good{98.3\%} & 1.049 & \better{0.223}
& \good{97.2\%} & 1.011 & 0.228 \\

Setting 1
& \bad{93.5\%} & 0.238 & 0.063
& \good{94.8\%} & 0.301 & 0.075
& \good{95.2\%} & 0.250 & 0.062
& \bad{88.9\%} & 0.169 & 0.054
& \good{94.8\%} & 0.237 & \better{0.061} \\

Setting 2
& \bad{87.9\%} & 0.330 & 0.100
& \bad{94.0\%} & 0.313 & 0.080
& \good{97.3\%} & 0.389 & \better{0.090}
& \bad{83.9\%} & 0.468 & 0.178
& \good{97.3\%} & 0.389 & \better{0.090} \\

Setting 3
& \good{94.9\%} & 0.252 & 0.063
& \good{95.2\%} & 0.300 & 0.074
& \good{99.6\%} & 0.290 & \better{0.050}
& \good{99.9\%} & 0.421 & 0.061
& \good{99.6\%} & 0.290 & \better{0.050} \\

Setting 4
& \good{96.2\%} & 0.235 & 0.058
& \good{95.2\%} & 0.300 & 0.074
& \good{99.3\%} & 0.276 & 0.052
& \good{100.0\%} & 0.368 & 0.048
& \good{96.1\%} & 0.178 & \better{0.045} \\
\bottomrule
\end{tabular}
}
\caption{Comparison of PLRD with other alternatives:~(a) and (b) correspond to inflating conventional LLR and LQR confidence intervals by multiplying the standard error by $2.181$ and $2.113$, respectively  \citep{ArmstrongKolesar2020}; (c) and (d) correspond to PLRD minus \cref{assu:part_lin} or \cref{assu:smooth}, respectively. Other details of this table are same as for \cref{tab:wgan} in the main text.}
\label{tab:wgan-ablation}
\end{table}

\begin{enumerate}[label=(\alph*)]
    \item  Bias-aware intervals for local linear regression with inflated critical values as proposed by  \citep{ArmstrongKolesar2020}:\footnote{Our end-to-end implementation differs from that in \citep{ArmstrongKolesar2020}, as we choose bandwidths for local linear regression using the bandwidth
    selector of \citep{imbens2012optimal} as implemented in \texttt{rdrobust}; whereas \citep{ArmstrongKolesar2020}
    discuss adaptive bandwidth choice using an alternative rule of thumb. Here, we chose to pick bandwidths as
    in \texttt{rdrobust} to enable a cleaner comparison to the conventional estimator.}
    \begin{enumerate}[label=(\roman*)]
        \item Fit a local linear regression using the \texttt{rdrobust} package, with default bandwidth selection rule.
        \item Obtain the LLR estimator $\wh\tau_\texttt{LLR}$ with its conventional standard error estimate $\wh{\text{SE}}_\texttt{LLR}$ 
        \item Build the confidence interval $\wh\tau_\texttt{LLR} \pm 2.181\cdot \wh{\text{SE}}_\texttt{LLR}$.
    \end{enumerate}
    \item Bias-aware intervals for local quadratic regression with inflated critical values as proposed by \citep{ArmstrongKolesar2020}:
    \begin{enumerate}[label=(\roman*)]
        \item Fit a local quadratic regression using the \texttt{rdrobust} package (i.e., with $p=2$ in \cref{sec:rdrobust}), with default bandwidth selection rule.
        \item Obtain the LQR estimator $\wh\tau_\texttt{LQR}$ (defined as $\wh\tau_2(h_n)$ in \cref{sec:rdrobust}) with its conventional standard error estimate $\wh{\text{SE}}_\texttt{LQR}$ 
        \item Build the confidence interval $\wh\tau_\texttt{LQR} \pm 2.113\cdot \wh{\text{SE}}_\texttt{LQR}$.
    \end{enumerate}
    \item PLRD without \cref{assu:part_lin}: Bias-aware inference operating under the class $$\cal{M}_{\text{H3},\, B}=\left\{\mu_{w}(\cdot): \mu''_{w}(x) \text{ is $M$-Lipschitz for all $x$ and } w=0,1\right\}.
$$ This method uses the additional smoothness just as PLRD, while allowing different curvatures around the treatment threshold. 
    \begin{enumerate}[label=(\roman*)]
        \item Estimate $B$ using a global cubic regression on either side of the threshold, in the same manner as we estimate $B$ via global cubic regression in \cref{prop:B}.
        
        \item Follow \citep{ArmstrongKolesar2018} and \citep{optrdd2019} to solve the minimax problem \eqref{eqn:minimax-problem0} with the function class $\cal{M}_{\text{H3},\, \wh{B}}$ via numerical convex optimization and compute the minimax linear estimator for $\t$ within this new class.

        \item As a by-product of the optimization problem, obtain a bound on the worst-case bias of our estimator within this class, and construct ``bias aware" confidence intervals for $\t$ using the \citep{ImbensManski2004} construction. 
    \end{enumerate} 
    Algorithmically, this is similar to our \cref{algo:PLRD}; we also  implemented this as a part of our \texttt{plrd} package.
    
    \item PLRD without \cref{assu:smooth}: Bias-aware intervals operating under the intersection of the class $\cal{M}_{\text{H2},\, M}$ and the class of regression functions satisfying \cref{assu:part_lin}. In essence, this method combines \texttt{optrdd} \citep{optrdd2019} with \cref{assu:part_lin}.
    \begin{enumerate}[label=(\roman*)]
        \item Here we estimate $M$ using a global quadratic regression on either side of the threshold, in the same manner as we estimate $B$ via global cubic regression in \cref{prop:B}.
        \item Follow \citep{ArmstrongKolesar2018} and \citep{optrdd2019} to solve the minimax problem \eqref{eqn:minimax-problem0} with the above function class via numerical convex optimization and compute the minimax linear estimator for $\t$ within this new class.
        
        \item Obtain an upper bound on the worst-case bias of our estimator within this class, and construct ``bias aware" confidence intervals for $\t$ using the \citep{ImbensManski2004} construction. 
    \end{enumerate}
    Algorithmically, this is similar to the \texttt{optrdd} package by \citep{optrdd2019}.\footnote{In fact, over the weaker smoothness class $\cal{M}_{\text{H2},\, M}$, \cref{assu:part_lin} on partial linearity becomes much less useful, because the second derivative of the baseline effect can have (bounded) sudden jumps anyways. As such, we also expect numerical performance of this ablation baseline to be close to \texttt{optrdd}.} We also provide an implementation of this method as a part of our replication files.
\end{enumerate}

The empirical results demonstrate that the rule-of-thumb methods (AK20 (a) and AK20 (b)) achieve valid coverage but at the cost of wider intervals, mirroring our earlier findings in \cref{tab:wgan,tab:simulations}. The relaxed variants of PLRD offer insights into the role of each assumption: Dropping either \cref{assu:smooth} or \cref{assu:part_lin} generally leads to more conservative intervals, suggesting that both assumptions contribute to the efficiency of PLRD.

\section{Additional Details on Real Data Experiments}\label{app:real-data-details}

Here we provide detailed descriptions of the real datasets analyzed in the main paper. We reference each dataset by the paper that introduced it, same as in \cref{tab:wgan,tab:real-data-CIs,tab:wgan-ablation}.

\begin{enumerate}
    \item \citep{lee2008}: This dataset, frequently used for illustrating  regression discontinuity methods, examines the incumbency advantage in U.S. House elections. The running variable is the margin of victory, defined as the difference in vote share between the Democratic candidate and their strongest opponent, with the threshold set at $c = 0$. The outcome of interest is the probability of winning subsequent elections. The dataset consists of $6{,}558$ observations from U.S. House elections between 1946 and 1998.  \citep{lee2008} reported a significant positive effect of incumbency on the probability of winning subsequent elections, indicating an incumbency advantage; see Table 1 and Figure 2 of \citep{lee2008} for more details.

    \item \citep{cattaneo2015senate}: This is similar to \citep{lee2008} except that they study the incumbency advantage in U.S.~Senate elections instead of House elections.  The running variable is the margin of victory, defined as the difference in vote percentages (ranging from $-100$ to $100$) between the Democratic candidate and their strongest opponent, with the threshold set at $c = 0$. The outcome of interest is the percentage of winning subsequent elections. The dataset consists of $1{,}390$ observations ($1{,}297$ excluding missing outcomes) from U.S. House elections between 1914 and 2010. Our findings align with  the original paper \citep{cattaneo2015senate} as well as the replication in \citep{cattaneo2024book}.

    \item \citep{matsudaira}:  
This dataset records fifth-grade students' math and reading scores in an unspecified large school district in the northeast from 2001 to 2002. \citep{matsudaira} examined the impact of mandatory summer school on students' performance. The running variable is the minimum of students' test scores in {math} and {reading}, with a score below the passing threshold for either math or reading determining mandatory attendance in the summer school. The primary outcomes are post-summer school reading and math test scores.
We use a version of this dataset shared by the author with K.~Kalyanamaran and G.~Imbens in 2008. This dataset has a slightly larger sample size $n = 68{,}798$
than the one used in the original paper ($n = 66{,}839$).
Our findings  mirror the results in \citep[Table 4]{matsudaira} for fifth-grade students.

    \item  \citep{jacob2004}: This dataset contains the record of third-grade student achievement in Chicago Public Schools from 1997 to 1999. \citep{jacob2004} investigated the effectiveness of remedial education programs on student achievement, using a regression discontinuity design based on test score. We use the same dataset as
\citep{imbens2009regression}, which contains $70{,}831$ observations after dropping entries with missing information. The running
variable is the minimum of the math and reading score for an exam in the
spring and the outcome  is the post-remedial program reading (resp.~math) test scores, reported as Rasch test scores rather than grade equivalents. Rasch test
    scores are normalized scores (zero mean, unit s.d.) for the population of third-grade students from 1997 to 1999. Our results are aligned with the findings in \citep{jacob2004} for third-grade students.

    \item \citep{Oreopoulos}: This study concerns the impact of compulsory schooling laws on educational attainment. The running variable is the individual’s year of birth, with thresholds determined by changes in compulsory schooling laws. The outcome here is years of completed education. 
Our dataset comprises $73{,}954$ observations, differing slightly from the original study due to variations in preprocessing; again, we use the same dataset also used in \citep{KolesarRothe2018} and focus on the reduced form analysis rather than a fuzzy RD approach (see Table 1, Panel C of the \href{https://assets.aeaweb.org/asset-server/articles-attachments/aer/contents/corrigenda/corr_aer.96.1.152.pdf}{corrigendum} to \citep{Oreopoulos}). Our 
findings are in line with Table 2 of \citep{KolesarRothe2018}.

    \item \citep{LudwigMiller}: This study examines the effectiveness of the Head Start program on children's health and schooling.  The running variable is the 1960 poverty rate, and the threshold corresponds to the qualifying cutoff for program participation. Following \citep{CCT2014robustCI}, we consider as outcome variable the mortality rates per 100,000 for children between 5 and 9 years old, with Head Start-related
causes, for 1973–1983. After cleaning missing values, this dataset includes $2{,}783$ observations. Table III and Figure IV (Panel A) in \citep{LudwigMiller} illustrate that the Head Start program had a negative impact on the outcomes.

\item \citep{meyersson2014islamic}: This study  investigates the effect of Islamic party control in local governments on women’s rights, in particular on the educational attainment of young women, focusing on Turkey's 1994 mayoral elections. The study employs an RD design where the running variable is the Islamic margin of victory—the vote share gap between the leading Islamic and secular parties, thus with a natural cutoff of zero. The analysis is conducted at the municipality level and has $2{,}629$ observations. We replicate part of the study 
where the primary outcome is the school attainment for women who were (potentially) in high school during the period 1994--2000, measured with variables extracted from the 2000 census.  This data have been used as a running example in \citep{cattaneo2024book}. Our results align with the illustrations in \citep{cattaneo2024book} and the original paper \citep{meyersson2014islamic} --- in municipalities where the Islamic party barely won, the educational attainment of women is roughly 3\% higher than in municipalities where the party barely lost.

  \item  \citep{Akhtari2022}: This paper studies the effect of political turnover in Brazilian mayoral elections on the quality of public services---in particular, on municipal education. The running variable is the incumbent party's vote margin (the vote-share gap between the incumbent party and its strongest opponent), with a natural cutoff at zero. We run the WGAN experiments with three outcomes:~The 4th-grade and 8th-grade student test scores following \cite{noack2021flexible}, and headmaster replacement in municipal schools following \cite{calonico2025treatment}.
\end{enumerate}

\section{Implementation via Convex Optimization}\label{sec:implementation} 

In this section, we discuss how one can solve the optimization problem \eqref{eqn:minimax-problem0} using readily available software described in, {\it e.g.}, \citep{BoydVandenberghe}. Essentially, our task
is to solve quadratic minimization problems of the form:\footnote{Here,
the $\sigma_i^2$ should simply be taken as observation-level
weights in an optimization problem. If we set $\sigma_i^2 =\var(Y_i\mid X_i)$ to be the
conditional variance of the response $Y_i$ given the running variable $X_i$, then the
problem below is exactly the minimax linear estimation problem.
When we run this the optimization as specified in
Algorithm \ref{algo:PLRD}, we set $\sigma_i^2 = \wh \sigma^2$ via cross-fitting, as
described in Section \ref{sec:introduce-PLRD}.}
\begin{equation}\label{eqn:minimax-problem}
  \mmz_{\gamma} \left[\sup_{\mu_{w}\,\in\,\cal{M}_B} \left(\sum_{i=1}^n\g_i\mu_{W_i}(X_i)- (\mu_1(c) - \mu_0(c))\right)^2 + \sum_{i=1}^n  \sigma_i^2\g_i^2\right].
\end{equation}
Here, we take $\mathcal{M}_B$ of all conditional response functions $\mu_{w}(\cdot)$ satisfying \cref{assu:part_lin} and such that $\mu_0''(x)$ is globally $B$-Lipschitz.
Since $\g_i$ are measurable functions of $\{X_i\}_{i=1}^n$, the bias of the estimator $\wh\tau(\gamma) = \sum_{i=1}^n \g_i Y_i$ conditional on the running variable is given by
\begin{align*}
    \E \left[\sum_{i=1}^n \g_i Y_i \,\bigg|\, \{X_i\}_{i=1}^n \right] -\t &= \sum_{i=1}^n \g_i \mu_{W_i}(X_i)  -\t\\&= \sum_{i=1}^n \g_i (W_i(\t + \beta (X_i-c)) +  \mu_{0}(X_i))  -\t\\
    &= \sum_{i=1}^n \g_i \left(\mu_{0}(X_i)-\mu_{0}(c)- \mu_{0}'(c) (X_i-c)-\frac{1}{2} \mu_{0}''(c) (X_i-c)^2\right),
\end{align*}
provided $$\sum_{i=1}^n \g_i=0,\ \sum_{i=1}^n \g_i W_i = 1,\ \sum_{i=1}^n \g_i W_i(X_i-c)=0,\ \sum_{i=1}^n \g_i (1-W_i)(X_i-c)=0,\ \sum_{i=1}^n \g_i (X_i-c)^2=0,$$ 
which are implicit constraints in our optimization procedure. Indeed, at values of $\g$ for which these constraints are not all satisfied, we can construct $\mu_{w}(\cdot)\in \mathcal{M}_B$  that forces the squared conditional bias to blow up to $\infty$. Thus, the solution $\wh\g$ to \eqref{eqn:minimax-problem} must satisfy the constraints. Next, define $$\rho(x):=\frac{1}{B}\left(\mu_{0}(x+c)-\mu_{0}(c)- \mu_{0}'(c) x-\frac{1}{2} \mu_{0}''(c) x^2\right).$$ 
Note that for any $\mu_{w}(\cdot)\in \mathcal{M}_B$, the function $\rho(\cdot)$ belongs to the function class $$\rhos:=\{f:f(0)=f'(0)=f''(0)=0, f'' \text{ is } \mtd\text{-Lipschitz}\}.$$
We can therefore rewrite the optimization problem \eqref{eqn:minimax-problem} as:
\begin{equation*}
     \wh{\g} := \mmz_{\gamma} \left[B^2\left(\sup_{\rho\in \rhos} \sum_{i=1}^n \g_i \rho(X_i-c)\right)^2 +\sum_{i=1}^n \sigma_i^2\g_i^2\right].
\end{equation*}
Finally, we make the implicit constraints explicit and use an auxiliary variable $t$ and rewrite the above optimization problem as:
\begin{align*}
(\wh\g,\wh{t}\,) := & \mmz_{\gamma, t}\left\{B^2t^2+\sum_{i=1}^n \g_i^2 \sigma_i^2\right\}\\
&\text{subject to }\ \sup_{\rho\in\rhos}\sum_{i=1}^n \g_i \rho(X_i-c)\le t,\ \sum_{i=1}^n \g_i W_i=1,\ \sum_{i=1}^n \g_i\left(1-W_i\right)=-1, \label{eqn:main-cvx-prob}\numberthis\\
&\qquad\qquad\ \sum_{i=1}^n\g_i (X_i-c)=0,\ \sum_{i=1}^n \g_i(1-W_i) (X_i-c)=0,\ \sum_{i=1}^n \g_i (X_i-c)^2=0.
\end{align*}
An immediate consequence of this modified optimization problem is that we get an upper bound on the conditional bias of our PLRD estimator. 
Next, we discuss how one can solve the above problem using convex duality. The dual problem is given by
\begin{align*}
& \mxz_{\rho(\cdot), \lambda} \inf _{\gamma, t}\left\{\sum_{i=1}^n \g_i^2 \sigma_i^2+B^2 t^2\right.\\
&\qquad\qquad\qquad\quad+\lambda_1\left(\sum_{i=1}^n \g_i \rho(X_i-c)-t\right) +\lambda_2\left(\sum_{i=1}^n \g_i W_i-1\right)+\lambda_3\left(\sum_{i=1}^n \g_i\left(1-W_i\right)+1\right)\\
&\qquad\qquad\qquad\quad  \left.+\lambda_4\sum_{i=1}^n \g_i W_i (X_i-c) +\lambda_5 \sum_{i=1}^n \g_i \left(1- W_i\right)(X_i-c)+\lambda_6\sum_{i=1}^n \g_i (X_i-c)^2\right\} \\
& \text {subject to } \lambda_1 \geq 0, \lambda_2, \lambda_3, \lambda_4, \lambda_5,\lambda_6 \in \R,\text{and},\rho\in\rhos.
\end{align*}
The inner minimization problem is simply a quadratic one, which we can explicitly solve to reduce the dual problem to the following.
\begin{equation}\label{dual:empirical}
    \wh{\theta}=\argmin_{\theta=(\wt{\rho},\lambda)\in \Theta} \left(\frac{1}{4n}\sum_{i=1}^n \frac{G^2(\theta; X_i-c, W_i)}{\sigma_i^2}+\frac{\lambda_1^2}{4nB^2}+\frac{\lambda_2-\lambda_3}{n}\right),
\end{equation}
where $\theta=(\wt{\rho},\lambda)$ denotes the dual variable that varies in the set
\begin{equation}\label{dual-param-space-duplicate}
    \Theta:=\left\{(f,\lambda)\mid f:[-\ell,\ell]\to\R, f(0)=f'(0)=f''(0)=0, f'' \text{ is $\lambda_1$-Lipschitz}, \lambda\in [0,\infty)\times \R^5\right\},
\end{equation}
and 
\begin{equation*}
    G(\wt{\rho},\lambda; x, w):=\wt{\rho}(x)+\lambda_2 w+\lambda_3(1-w)+\lambda_4 wx +\lambda_5 (1-w)x+\lambda_6 x^2.
\end{equation*}
The function $\wt{\rho}(\cdot)$ in the  above formulation corresponds to the function $\lambda_1 \rho(\cdot)$ in the original problem. Thus, we can recover the optimal variables $(\wh\g,\wh{t}\,)$ for the primal problem \eqref{eqn:main-cvx-prob} from the dual optimal variable $\wh\theta$ (i.e., solution of \eqref{dual:empirical}), as follows.
$$\wh\g_{i} = -\frac{1}{2\sigma_i^2}G(\wh\theta; X_i-c, W_i),\quad \wh{t} = \frac{1}{2B^2}\wh{\lambda}_1.$$
Since our optimization problem differs from \citep{optrdd2019} essentially in the choice of the function class, \citep[Proposition 2]{optrdd2019} provides a basis for approximating \eqref{dual:empirical} using a finite-dimensional quadratic program, in a way that the resulting discrete solutions converge uniformly to the continuous solution as the grid size becomes small. In our software implementation, we adopt this finite-dimensional approximation approach to compute the optimal weights using \texttt{quadprog}, a function for quadratic programming in \texttt{R} \citep{r2023}.

\section{Supporting Lemmas}\label{app:aux-results}

To set the stage, we recall the key concepts of entropy with and without bracketing for function classes, which will be central to the results presented in this section.
Consider a class of real-valued functions $\mc{F}$ defined on a set $\mc{Z}$, equipped with a norm $\|\cdot\|$. The \emph{covering number} $N(\eps, \mc{F}, \|\cdot\|)$ is the minimum number of balls of radius $\eps$ needed to cover $\mc{F}$. The centers of the covering balls need not belong to $\mc{F}$, but they are required to have finite norms. The \emph{entropy (without bracketing)} is defined as the logarithm of the covering number.
An $\eps$-bracket is a pair of functions $(l, u)$ such that $\|u-l\| < \eps$ and $l \leq f \leq u$ for all $f \in \mathcal{F}$. The \emph{bracketing number} $N_{[\, ]}(\eps, \mc{F}, \|\cdot\|)$ is the minimum number of $\eps$-brackets needed to cover $\mc{F}$.  The upper and lower bounds $u$ and $l$ of the brackets need not belong to $\mc{F}$, but they are required to have finite norms. The \emph{entropy with bracketing} is defined as the logarithm of the bracketing number.

Following~\citep[Section 2.7]{VW},  we consider the class of functions $f: \mc{Z} \to \R$ defined on a bounded set $\mc{Z} \subset \R$ such that
the derivatives of $f$ up to order $k$ are uniformly bounded, and
$f^{(k)}$ is Lipschitz continuous.  We can define a norm on this class as follows:
$$\|f\|_{k+1} := \max_{j \leq k} \sup_x |f^{(j)}(x)| + \sup_{x,y} \frac{|f^{(k)}(x) - f^{(k)}(y)|}{|x-y|},$$
where the last supremum is over all $x, y$ in the interior of $\mc{Z}$ with $x \neq y$. Finally, denote by $\mc{C}_M^k(\mc{Z})$ the set of all continuous functions $f: \mc{Z} \to \R$ with $\|f\|_{k} \leq M$. Throughout, we use $A \lesssim_\eta B$ to denote that $A\le CB$ for some arbitrary constant $C>0$ that can depend on $\eta$. For a vector $\lambda=(\lambda_{1},\lambda_{2},\dots,\lambda_{p})$ we denote $\lambda_{-1}:=(\lambda_{2},\dots,\lambda_{p})$.  We use $\|x\|_\infty$ to denote $\max_{i} |x_i|$ when $x$ is a vector, and $\sup_{t} |x(t)|$ when $x$ is a function.

\begin{lemma}\label{approx-class-entropy}
Consider the following class of functions:
$$\mc{F}_{B,\,\ell,\,k}=\{\rho\mid \rho:[-\ell,\ell]\to \R, \ \rho(0)=\rho'(0)=\cdots=\rho^{(k)}(0)=0, \ \rho^{(k)} \text{ is $B$-Lipschitz}\},$$
where $k\ge 0$ is an integer and $0<\ell\le 1$.
It holds that $$\log N(\eps, \mc{F}_{B,\,\ell,\,k}, \|\cdot\|_\infty)\lesssim \ell (B/\eps)^{1/(k+1)} ,\quad\text{and}\quad \log N_{[\,]}(\eps, \mc{F}_{B,\,\ell,\,k}, L_r(P))\lesssim \ell (B/\eps)^{1/(k+1)} ,$$
for any $r\ge 1$ and probability measure $P$ on $\R$.
Thus, $\mc{F}_{B,\,\ell,\,k}$ is universally Donsker. 
\end{lemma}

\begin{proof}[Proof of~\cref{approx-class-entropy}]
    For any $\rho\in \mc{F}_{B,\,\ell,\,k}$, the derivatives of $\rho$ upto order $k$ are all bounded by $B \ell $. Thus, $\|\rho\|_{k+1}\le B(1+\ell)\le 2B$.
    Consequently,\begin{align*}
        N(\eps, \mc{F}_{B,\,\ell,\,k}, \|\cdot\|_\infty)&\le N(\eps, \mc{C}^{k+1}_{2B}([-\ell,\ell]), \|\cdot\|_\infty)\\
        &\le N(\eps/2B, \mc{C}^{k+1}_1([-\ell,\ell]), \|\cdot\|_\infty).
    \end{align*}
   For any $\delta>0$, note that $N(\delta, [-\ell,\ell], |\cdot|)\lesssim \ell/\delta$. We can thus slightly tweak the proof of \citep[Theorem 2.7.1]{VW}, with $\alpha=k+1$ and $\lambda(\cal{X}_1)=2\ell$ (in their notation), to conclude that 
    $$\log N(\eps/2B, \mc{C}^{k+1}_1([-\ell,\ell]), \|\cdot\|_\infty))\lesssim \ell (\eps/2B)^{-1/(k+1)},$$
   as desired to show. This gives us the first entropy bound (without bracketing). 
   
   Next, we can adapt the proof technique of \citep[Corollary 2.7.2]{VW} to deduce that the last display also yields the desired bound on entropy with bracketing. Finally, the bound on $\log N_{[\,]}(\eps, \mc{F}_{B,\,\ell,\,k}, L_2(P))$ implies that $\mc{F}_{B,\,\ell,\,k}$ is $P$-Donsker for any probability distribution $P$ (see, {\it e.g.,} \citep[Theorem 19.5]{vdV}).
\end{proof}

For the rest of this section, we assume without loss of generality that the treatment threshold is at $c=0$.
Before stating the next few lemmas, we recall some notation. The empirical dual problem is given by
$$\wh{\theta}_n(\sigma^2,B)=\argmin_{\theta=({\rho},\lambda)\in \Theta(\ell)} \wh{L}_n(\theta; B), \quad \wh{L}_n(\theta; B):= \frac{1}{4 n}\sum_{i=1}^n \frac{G^2(\theta; X_i, W_i)}{\sigma_i^2} +\frac{\lambda_1^2}{4nB^2}+\frac{\lambda_2-\lambda_3}{n},$$ where
$G(\theta; x, w)=G(\rho, \lambda;x,w)={\rho}(x)+\lambda_2 w+\lambda_3(1-w)+\lambda_4 wx +\lambda_5 (1-w)x+\lambda_6 x^2$, with $w=1(x>0)$,
and $$\Theta(\ell):=\left\{(f,\lambda)\mid f:[-\ell,\ell]\to \R, f(0)=f'(0)=f''(0)=0, f'' \text{ is $\lambda_1$-Lipschitz}, \lambda\in [0,\infty)\times \R^5\right\}.$$
Define the dual problem at the population level as
$$\theta_n^*(\sigma^2,B) = \argmin_{\theta\in \Theta(\ell)}L_n(\theta;\sigma^2,B), \quad \text{where}\quad L_n(\theta;\sigma^2,B):=\frac{1}{4\sigma^2}\E\left[G^2(\theta; X, W)\right]+\frac{\lambda_1^2}{4nB^2}+\frac{\lambda_2-\lambda_3}{n}.$$
In \cref{algo:PLRD}, we sample splitting to estimate the unknown Lipschitz constant $B$ and the homoskedastic variance $\sigma^2=\var(Y_i\mid X_i)$. However, for ease of exposition, we assume throughout this proof that an estimate $\wh{B}$ of $B$ and an estimate $\wh\sigma^2$ of $\sigma^2$ is available independently of the data using which we solve the empirical dual problem. We define
$$\wh\theta_n = \wh\theta_n(\wh\sigma^2,\wh{B}),\quad\text{and}\quad\theta_n^*=\theta_n^*(\wh\sigma^2,B).$$

 \begin{lemma}\label{lemma:uniform-concentration} Define  $$\Theta_n(\ell,\delta,\eps):=\{(\rho,\lambda)\mid (\rho,\lambda)\in \Theta(\ell),\ |\lambda_1-\lambda_{n,1}^*|\le \delta \lambda_{n,1}^* \text{ and } \|\lambda_{-1}-\lambda^*_{n,-1}\|_2\le \eps \|\lambda_{n,-1}^*\|_2\},$$ where  $\ell,\delta,\eps\in(0,1)$ are fixed. Then, the following maximal inequality holds.
$$\E \sup_{\theta\,\in\,\Theta_n(\ell,\delta,\eps)} |\wh{L}_n(\theta;\wh\sigma^2,\wh{B})-L_n(\theta;\wh\sigma^2,\wh{B})|\lesssim (\lambda_{n,1}^*)^2\frac{\ell^6}{\sqrt{n}}+\|\lambda_{n,-1}^*\|_2^2\frac{\sqrt{\eps}}{\sqrt{n}}+\lambda_{n,1}^*\|\lambda_{n,-1}^*\|_2\frac{\ell^3}{\sqrt{n}}.$$
\end{lemma}

\begin{proof}[Proof of~\cref{lemma:uniform-concentration}]
We begin by noting that $$\left|\wh{L}_n(\theta;\wh\sigma^2,\wh{B})-L_n(\theta;\wh\sigma^2,\wh{B})\right|=\frac{1}{4\wh\sigma^2}|(\P_n-\P)G^2(\theta; X, W)|.$$ We prove a maximal inequality for the above RHS using tools from empirical process theory. Write $G(\rho,\lambda; x, w)={\rho}(x)+\lambda_{-1}^\top v(x,w)$ where $v(x,w):=(w,1-w,wx,(1-w)x,x^2)^\top$, and note that $\|v(x,w)\|_2\le \sqrt{1+\ell+\ell^2}\le 2$. To reduce notation, we write $\Theta_n=\Theta_n(\ell,\delta,\eps)$ for the rest of this proof.  Using triangle inequality,
\begin{align*}
    \sup_{\theta\in \Theta_n} |(\P_n-\P)G^2(\theta;X,W)|&\le \sup_{(\rho,\lambda)\in \Theta_n} |(\P_n-\P)\rho^2(X)|+ \sup_{(\rho,\lambda)\in \Theta_n} |(\P_n-\P)(\lambda_{-1}^\top v(X,W))^2|\\
    &\qquad+ 2\sup_{(\rho,\lambda)\in \Theta_n} |(\P_n-\P)\rho (X)\lambda_{-1}^\top v(X,W)|.
\end{align*}
We will bound each of the above three quantities in order. First, note that for any $(\rho,\lambda)\in \Theta_n$, $\rho^{(j)}(0)=0$ for $j=0,1,2$, and $\rho''$ is $B_1:=(1+\delta)\lambda_{n,1}^*$-Lipschitz. This gives $|\rho''(x)|=|\rho''(x)-\rho''(0)|\le B_1 \ell$. Moreover, a Taylor expansion of ${\rho}(x)$ about $x=0$ yields the following 
\begin{equation}\label{Taylor-expand-rho}
    |{\rho}(x)|=|\rho(x)-\rho(0)-\rho'(0)x|=\frac{|\rho''(\xi)-\rho''(0)|}{|\xi - x|}\left|\xi-x\right|\frac{x^2}{2}\le B_1\frac{\ell^3}{2}.
\end{equation}
A similar argument gives $|\rho'(x)|\le B_1 \ell^2$. Putting these all together, we can say that the function $r=\rho^2$ satisfies $r(0)=0$, $r'(0)=0$, $r''(x)=2\rho''(x)\rho(x)+2\rho'(x)^2$, and thus $$\|r''\|_\infty \lesssim B_1^2 \ell^4 \lesssim (\lambda_{n,1}^*)^2\ell^4.$$
Consequently, using the notation of \cref{approx-class-entropy}, we can say that $\mc{A}:=\{\rho^2:(\rho,\lambda)\in \Theta_n\text{ for some }\lambda\}$ is a subset of the class $\mc{F}_{\ell,h,1}$ with Lipschitz constant $B\lesssim (\lambda_{n,1}^*)^2\ell^4$. Therefore, \cref{approx-class-entropy} implies that $$\log N_{[\,]}(\eps', \cal{A}, L_2(P))\le \log N_{[\,]}(\eps', \mc{F}_{\ell,h,1}, L_2(P))\le \ell (B/\eps')^{1/2}\lesssim \lambda_{n,1}^* \ell^3 (\eps')^{-1/2}.$$
Finally, note that \eqref{Taylor-expand-rho} implies that the constant function $F=(\lambda_{n,1}^* \ell^3)^2$ works as an envelope to the class $\cal{A}$.
We can now apply \cite[Corollary 19.35]{vdV} to conclude that
\begin{align*}
    \E\sup_{(\rho,\lambda)\in \Theta_n} |(\P_n-\P)\rho^2(X)| &= \E\sup_{r\in\mc{A}}|(\P_n-\P)r|\\
    &\lesssim \frac{1}{\sqrt{n}}\int_0^{\|F\|_{P,2}} \sqrt{\log N_{[\,]}(\eps', \mc{A}, L_2(P))}\,d\eps' \\
    &\lesssim \frac{1}{\sqrt{n}}\int_0^{F} (\eps')^{-1/4}\left(\lambda_{n,1}^* \ell^3\right)^{1/2} \,d\eps'\\
    &\lesssim\frac{1}{\sqrt{n}} F^{3/4}F^{1/4}=(\lambda_{n,1}^*)^2\frac{\ell^6}{\sqrt{n}}.
\end{align*}
Next, we deal with $$\sup_{(\rho,\lambda)\in \Theta_n} |(\P_n-\P)(\lambda_{-1}^\top v(X,W))^2|=\sup_{\lambda\in\Lambda_n} |(\P_n-\P)(\lambda_{-1}^\top v(X,W))^2|,$$
where $\Lambda_n=\Lambda_n(\eps)=\{v\in\R^5:\|v-\lambda_{n,-1}^*\|\le\eps\|\lambda_{n,-1}^*\|\}\subseteq [0,\infty)\times \R^5$. Note, for any $\lambda\in\Lambda_n$,
$$|\lambda_{-1}^\top v(x,w)|\le \|\lambda_{-1}\|_2\|v(x,w)\|_2\le (1+\eps)\|\lambda_{n,-1}^*\|_2\sqrt{1+\ell+\ell^2}\le 4\|\lambda_{n,-1}^*\|_2.$$
Therefore for any fixed $x,w$, and $\lambda,\wt{\lambda}\in \Lambda_n$,
\begin{align*}
    |\langle\lambda, v(x,w))\rangle^2-\langle\wt{\lambda}, v(x,w))\rangle^2|&\le 8\|\lambda_{n,-1}^*\|_2|\langle\lambda, v(x,w))\rangle-\langle\wt{\lambda}, v(x,w))\rangle|\\
    &\le 16\|\lambda_{n,-1}^*\|_2 \|\lambda-\wt{\lambda}\|_2,\end{align*}
showing that the function $\lambda\mapsto (\lambda^\top v(x,w))^2$ (where $\lambda\in \Lambda_n$) is $c$-Lipschitz, for $c=16\|\lambda_{n,-1}^*\|$. Consequently, 
\begin{align*}
\log N_{[\,]}(\eps', \Lambda_n,L_2(P))&\le 5\log\left(1+\frac{2c \diam(\Lambda_n)}{\eps'}\right)\\
&\le \frac{10c\diam(\Lambda_n)}{\eps'}\lesssim \eps\|\lambda_{n,-1}^*\|_2^2 (\eps')^{-1}.\end{align*}
We can now apply \cite[Corollary 19.35]{vdV} with envelop $(4\|\lambda_{n,-1}^*\|_2)^2$ to conclude that
\begin{align*}
    \E\sup_{(\rho,\lambda)\in \Theta_n} |(\P_n-\P)(\lambda_{-1}^\top v(X,W))^2| 
    &\lesssim \frac{1}{\sqrt{n}}\int_0^{16\|\lambda_{n,-1}^*\|^2} \sqrt{\log N_{[\,]}(\eps', \Lambda_n,L_2(P))}\,d\eps' \\
    &\lesssim  \|\lambda_{n,-1}^*\|_2^2\frac{\sqrt{\eps}}{\sqrt{n}}.
\end{align*}
It only remains to prove a maximal inequality for the following supremum:
$$\sup_{(\rho,\lambda)\in \Theta_n} |(\P_n-\P)\rho (X)\lambda_{-1}^\top v(X,W)|.$$
Recall from \eqref{Taylor-expand-rho} that $|\rho(x)|\le B_1 \ell^3/2\le \lambda_{n,1}^* \ell^3$ for any $(\rho, \lambda)\in \Theta_n$. Also recall from above that $|\langle\lambda_{-1}, v(x,w)\rangle|\le  4 \|\lambda_{n,-1}^*\|_2$. Consequently, for any $(\rho,\lambda)\in \Theta_n$, and for every $x,w$, we have 
$$
|\rho (x)\lambda_{-1}^\top v(x,w)|\le 4\lambda_{n,1}^*\|\lambda_{n,-1}\|_2 \ell^3,
$$
proving that the constant function $4\lambda_{n,1}^*\|\lambda_{n,-1}\|_2 \ell^3$ works as an envelope for this class.

Next, we fix $\eps'>0$. Suppose that $\cal{R}_\eps:=\{\rho_1,\dots,\rho_{N_1}\}$ is a minimal $\eps_1$-cover for  $\pi_1(\Theta_n)=\{\rho:(\rho,\lambda)\in\Theta_n\text{ for some }\lambda\}$ under the sup-norm, and $\cal{L}_\eps:=\{B_1,\dots,\ell_{N_2}\}$ is a minimal $\eps_2$-cover for $\Lambda_n=\{v\in \R^5:\|v-\lambda_{n,-1}^*\|_2\le \eps\|\lambda_{n,-1}^*\|_2\}$ under the $\ell_2$-norm, where $\eps_1 := {\eps'}/{8\|\lambda_{n,-1}^*\|_2}$ and $\eps_2={\eps'}/{4\lambda_{n,1}^*\ell^3}$. It is well-known (see, {\it e.g.}, \citep{vdV}) that,
$$\log N_2\le 5\log\left(1+\frac{c\eps\|\lambda_{n,-1}^*\|}{\eps'/4\lambda_{n,1}^*\ell^3}\right)\lesssim\eps\lambda_{n,1}^*\|\lambda_{n,-1}^*\|\ell^3(\eps')^{-1}.
$$ On the other hand, \cref{approx-class-entropy} yields $$\log N_1\lesssim (\eps'/8\|\lambda_{n,-1}\|_2)^{-1/3}(2\lambda_{n,1}^*)^{1/3}\ell = (\lambda_{n,1}^*\|\lambda_{n,-1}^*\|_2)^{1/3}\ell(\eps')^{-1/3}.$$  Now, for any $\theta=(\rho,\lambda)\in \Theta_n$, pick $\theta^0=(\rho^0,\lambda^0)\in \cal{R}_\eps\times\cal{L}_\eps$ such that $\|\rho-\rho^0\|_\infty\le \eps_1$ and  $\|\lambda-\lambda^0\|_2\le \eps_2$. Note that,
\begin{align*}
    &|\rho (x)\langle\lambda_{-1}, v(x,w)\rangle - \rho^0 (x)\langle\lambda_{-1}^0, v(x,w)\rangle|\\
    &\le |\rho (x)|\cdot |\langle\lambda_{-1}, v(x,w)\rangle - \langle\lambda_{-1}^0, v(x,w)\rangle|+|\langle\lambda_{-1}^0, v(x,w)\rangle|\cdot |\rho (x)- \rho^0 (x)|\\
     &\le \lambda_{n,1}^*\ell^3 \cdot\|\lambda_{-1} - \lambda_{-1}^0\|_2\|v(x,w)\|_2+4\|\lambda_{n,-1}^*\|_2\cdot |\rho (x)- \rho^0 (x)|\\
     &\le \lambda_{n,1}^*\ell^3\cdot \frac{\eps'}{4\lambda_{n,1}^*\ell^3}\cdot 2+4\|\lambda_{n,-1}^*\|_2 \cdot \frac{\eps'}{8\|\lambda_{n,-1}^*\|_2}\le \eps'.
\end{align*}
Therefore, the class $\mc{B}$ of functions $(x,w)\mapsto \rho(x)\langle\lambda_{-1}, v(x,w)\rangle$, where $(\rho,\lambda)\in\Theta_n$, satisfies
\begin{align*}
    \log N(\eps', \mc{B}, \|\cdot\|_\infty)
    &\le \log N_1N_2\lesssim (\eps')^{-1/3}(\lambda_{n,1}^*\|\lambda_{n,-1}^*\|_2)^{1/3}\ell + (\eps')^{-1}\lambda_{n,1}^*\|\lambda_{n,-1}^*\|\ell^3\eps.
\end{align*}
This also gives us a bound for the entropy with bracketing (cf.~\citep[Corollary 2.7.2]{VW}), as follows. 
\begin{align*}
    \log N_{[\,]}(\eps', G(\Theta_n), L_2(P))&\le \log N(2\eps', G(\Theta_n), |\cdot|_\infty)\\
    &\lesssim  (\eps')^{-1/3}(\lambda_{n,1}^*\|\lambda_{n,-1}^*\|_2\ell^3)^{1/3} + (\eps')^{-1}\lambda_{n,1}^*\|\lambda_{n,-1}^*\|\ell^3\eps.
\end{align*}
Recall from above that the constant function $F_{ab}:=4\lambda_{n,1}^*\|\lambda_{n,-1}^*\|_2 \ell^3$ acts as an envelop function for the class $\mc{B}$. We now apply \cite[Corollary 19.35]{vdV} once more to conclude that
\begin{align*}
    \E\sup_{(\rho,\lambda)\in \Theta_n} |(\P_n-\P)\rho (X)\lambda_{-1}^\top v(X,W)|&\lesssim \frac{1}{\sqrt{n}}\int_0^{F_{ab}} \left((\eps')^{-1/6}F_{ab}^{1/6}+(\eps')^{-1/2}\eps^{1/2}F_{ab}^{1/2}\right)\,d\eps'\\
    &\lesssim\frac{1}{\sqrt{n}}\left((F_{ab})^{5/6}(F_{ab})^{1/6} + (F_{ab})^{1/2}\eps^{1/2}(F_{ab})^{1/2}\right)\\
    &\lesssim \lambda_{n,1}^*\|\lambda_{n,-1}^*\|_2 \frac{\ell^3}{\sqrt{n}},
\end{align*}
which completes the proof.
\end{proof}

\begin{lemma}\label{lemma:strongly-cvx} 
Let $G(\rho,\lambda;x,w):=\rho(x) + \lambda_{-1}^\top v(x,w)$, where $v(x,w):=(w,1-w,xw, x(1-w), x^2)$. 
Assume that the assumptions in \cref{thm:validity} hold. Then, there exists $0<\alpha\le 1$ such that, for any $u\in \R^5$, $$u^\top \E \left[v(X,W)v(X,W)^\top \right]u\ge \alpha \|u\|_2^2.$$
Further, $\theta_n^*(\sigma^2,B)=\argmin_{\theta\,\in\,\Theta(\ell)} L_n(\theta;\sigma^2,B)$ satisfies the following.
\begin{align*}
    L_n(\theta;\sigma^2,B)-L_n(\theta_n^*;\sigma^2,B)&\ge\frac{1}{4\sigma^2}\E\left(\rho(X)-\rho_n^*(X)+\langle \lambda_{-1}-\lambda_{n,-1}^*,v(X,W)\rangle\right)^2+\frac{(\lambda_1-\lambda_{n,1}^*)^2}{4nB^2}\\
    &\ge \frac{\alpha}{4\sigma^2}\left(\|\rho-\rho_n^*\|^2_{L^2(P)}+\|\lambda_{-1}-\lambda_{n,-1}^*\|_2^2\right)+\frac{(\lambda_1-\lambda_{n,1}^*)^2}{4nB^2}-R_n,
\end{align*}
where $R_n=\frac{2}{4\sigma^2}\E|(\rho(X)-\rho_n^*(X))\langle \lambda_{-1}-\lambda_{n,-1}^*,v(X,W)\rangle|$. \end{lemma}

\begin{proof}[Proof of~\cref{lemma:strongly-cvx}]

Define $A:=\E[v(X,W)v(X,W)^\top]$. Since $A$ is positive semidefinite, it has non-negative eigenvalues. We first show that the smallest eigenvalue of $A$, which we denote as $\lambda_{\min}(A)$, is strictly positive. To prove this, assume to the contrary that $\lambda_{\min}(A)=0$. Then, there exists $u\in \R^5$, $u\neq 0$, such that $u^\top A u = 0$. We can conclude from $u^\top A u=\E (v(X,W)^\top u)^2 = 0$ that $v(X,W)^\top u=0$ with probability $1$. This implies that, either $X$ is a solution of $u_1 + u_3 x + u_5 x^2=0$ (when $w=\ind{x>0}=1$), or a solution of $u_2 + u_4 x +u_5 x^2 =0$ (when $w=\ind{x>0}=0$). But, $X$ having a finite support contradicts the assumption that the running variable has a strictly positive density at the cutoff. So we must have $\lambda_{\min}(A)>0$. Therefore, 
$$u^\top \E \left[v(X,W)v(X,W)^\top \right]u\ge \alpha \|u\|_2^2,$$
for any $u\in \R^5$, where $\alpha:=\min\{1,\lambda_{\min}(A)\}$.

For proving the second part, we write $L_n=L_n(\,\cdot\,;\sigma^2,B)$ for notational convenience. First, we simplify $L_n(\theta)-L_n(\theta_n^*)$ using the fact that  gradients of the loss $L_n$ at $\theta_n^*$ are zero. To remove the dependence of the domain on the parameter $\lambda_1$, we write
$$L_n(\rho,\lambda)= \frac{1}{4\sigma^2}\E \left(\lambda_1 \wt{\rho}(X)+\lambda_{-1}^\top v(X,W)\right)^2+\frac{\lambda_1^2}{4nB^2}+\frac{\lambda_2-\lambda_3}{n}=: \wt{L}_n(\wt{\rho}, \lambda),$$
where $\wt{\rho}(\cdot)=\rho(\cdot)/\lambda_1\in\rhos$. Since $\wt{L}_n$ is minimized at $(\wt{\rho}_n^*,\lambda_n^*)$, we can deduce the following.
\begin{enumerate}
    \item $\nabla_{\lambda_{-1}} \wt{L}_n(\wt{\rho}_n^*,\lambda_n^*)=0$. This implies that,
    \begin{equation}\label{grad-lambda-minus}
        \frac{1}{2\sigma^2}\E\left[\left(\lambda_{n,1}^*\wt{\rho}_n^*(X)+\langle\lambda_{n,-1}^*,v(X,W)\rangle\right)v(X,W)\right]+\frac{e_1-e_2}{n}=0.
    \end{equation}
    where $e_1=(1,0,0,0,0)$ and $e_2=(0,1,0,0,0)$. 
    \item $\nabla_{\lambda_{1}} \wt{L}_n(\wt{\rho}_n^*,\lambda_n^*)=0$. This implies that,
    \begin{equation}\label{grad-lambda-one}
        \frac{1}{2\sigma^2}\E\left[\left(\lambda_{n,1}^*\wt{\rho}_n^*(X)+\langle\lambda_{n,-1}^*,v(X,W)\rangle\right)\wt{\rho}_n^*(X)\right]+\frac{2\lambda_{n,1}^*}{4nB^2}=0.
    \end{equation}
    \item $\nabla_{\rho} \wt{L}_n(\wt{\rho}_n^*,\lambda_n^*)=0$. We can say that the convex function $t\mapsto\wt{L}_n(t\wt{\rho}+(1-t)\wt{\rho}_n^*,\lambda_n^*)$ is minimized at $t=0$, which gives us the following.
    \begin{equation}\label{grad-rho}
        \frac{2\lambda_{n,1}^*}{4\sigma^2}\E\left[\left(\wt{\rho}(X)-\wt{\rho}_n^*(X)\right)\left(\lambda_{n,1}^*\wt{\rho}_n^*(X)+\langle\lambda_{n,-1}^*,v(X,W)\rangle\right)\right]\ge 0.
    \end{equation}
\end{enumerate}
Equipped with \cref{grad-lambda-minus,grad-lambda-one,grad-rho}, we are ready to simplify $L_n(\theta)-L_n(\theta_n^*)=\wt{L}_n(\wt{\rho},\lambda) - \wt{L}_n(\wt{\rho}_n^*, \lambda_n^*)$, as follows.
\begin{align*}
   & L_n(\theta)-L_n(\theta_n^*)-\frac{1}{4\sigma^2}\E\left(\rho(X)-\rho_n^*(X)+\langle \lambda_{-1}-\lambda_{n,-1}^*,v(X,W)\rangle\right)^2\\
   &=\frac{2}{4\sigma^2}\E\left[(\rho-\rho_n^*)(X)\left(\rho_n^*(X)+\langle \lambda_{n,-1}^*,v(X,W)\rangle\right)\right]+\frac{\lambda_{1}^2-(\lambda_{n,1}^*)^2}{4nB^2}\\
   &\qquad + \frac{2}{4\sigma^2}\E\left[\langle \lambda_{-1}-\lambda_{n,-1}^*,v(X,W)\rangle\left(\rho_n^*(X)+\langle \lambda_{n,-1}^*,v(X,W)\rangle\right)\right]+\frac{\lambda_2-\lambda_3-\lambda_{n,2}^*+\lambda_{n,3}^*}{n}.
   \end{align*}
Using \eqref{grad-lambda-minus}, we deduce that the second line in the above display is zero. Next, set $\wt{\rho}=\rho/\lambda_1$, $\wt{\rho}_n^*=\rho_n^*/\lambda_{n,1}^*$. Continuing from above,
   \begin{align*}
   &=\frac{2}{4\sigma^2}\E\left[(\lambda_1\wt{\rho}(X)-\lambda_{n,1}^*\wt{\rho}_n^*(X))\left(\rho_n^*(X)+\langle \lambda_{n,-1}^*,v(X,W)\rangle\right)\right]+\frac{\lambda_{1}^2-(\lambda_{n,1}^*)^2}{4nB^2}\\[2mm]
    &=\frac{2\lambda_1}{4\sigma^2}\E\left[(\wt{\rho}(X)-\wt{\rho}_n^*(X))\left(\rho_n^*(X)+\langle \lambda_{n,-1}^*,v(X,W)\rangle\right)\right]\\[2mm]
    &\qquad\qquad+\frac{2}{4\sigma^2}\E\left[(\lambda_1-\lambda_{n,1}^*)\wt{\rho}_n^*(X)\left(\rho_n^*(X)+\langle \lambda_{n,-1}^*,v(X,W)\rangle\right)\right]+\frac{\lambda_{1}^2-(\lambda_{n,1}^*)^2}{4nB^2}\\[2mm]
    &\ge 0-\frac{2\lambda_{n,1}^*}{4nB^2}(\lambda_1-\lambda_{n,1}^*)+\frac{\lambda_{1}^2-(\lambda_{n,1}^*)^2}{4nB^2}\tag{using \eqref{grad-rho} and \eqref{grad-lambda-one}}\\[2mm]
    &\ge\frac{(\lambda_1-\lambda_{n,1}^*)^2}{4nB^2}.
\end{align*}
Therefore,
\begin{align*}
    L_n(\theta)&-L_n(\theta_n^*)\ge\frac{1}{4\sigma^2}\E\left(\rho(X)-\rho_n^*(X)+\langle \lambda_{-1}-\lambda_{n,-1}^*,v(X,W)\rangle\right)^2+\frac{(\lambda_1-\lambda_{n,1}^*)^2}{4nB^2}\\
    &\ge\frac{1}{4\sigma^2}\left[\E(\rho(X)-\rho_n^*(X))^2+(\lambda_{-1}-\lambda_{n,-1}^*)^\top A(\lambda_{-1}-\lambda_{n,-1}^*)\right]-R_n+\frac{(\lambda_1-\lambda_{n,1}^*)^2}{4nB^2}\\
    &\ge\frac{\alpha}{4\sigma^2}\left(\|\rho-\rho_n^*\|^2_{L^2(P)}+\|\lambda_{-1}-\lambda_{n,-1}^*\|_2^2\right)+\frac{(\lambda_1-\lambda_{n,1}^*)^2}{4nB^2}-R_n,
\end{align*}
as desired to show.

\end{proof}

\begin{lemma}\label{lemma:one-to-minus-one}
   The optimizer $\theta_n^*=(\rho_n^*, \lambda_n^*)$ of the dual population problem  \eqref{dual:population} satisfies $$\|\lambda_{n,-1}^*\|_2=O(\lambda_{n,1}^*\ell^3+n^{-1}).$$
   Consequently, \begin{equation*}
\sup_{\theta\,\in\,\Theta_n(\ell,\delta,\eps)} |{L}_n(\theta;\wh{\sigma}^2,\wh{B})-L_n(\theta;\wh{\sigma}^2, B)|\lesssim 
\frac{(\lambda_{n,1}^*)^2}{n}\left|\frac{1}{\wh{B}}-\frac{1}{B}\right|.
\end{equation*}
\end{lemma}

\begin{proof}[Proof of \cref{lemma:one-to-minus-one}]
    Recall from \eqref{grad-lambda-minus} in the proof of \cref{lemma:strongly-cvx} that  $$\frac{1}{2\sigma^2}\E\left[\left(\lambda_{n,1}^*\wt{\rho}_n^*(X)+\langle\lambda_{n,-1}^*,v(X,W)\rangle\right)v(X,W)\right]+\frac{e_1-e_2}{n}=0.$$
    Rearranging this equation yields $$A\lambda_{n,-1}^* = -\lambda_{n,1}^*\E\wt\rho_n^*(X)v(X,W) - 2\sigma^2\frac{(e_1-e_2)}{n}$$
    where $A:=\E[v(X,W)v(X,W)^\top]$, as defined in the proof of \cref{lemma:strongly-cvx}. Since $\|v(X,W)\|_2\le 2$ and \eqref{Taylor-expand-rho} yields $\|\wt\rho_n^*\|_\infty\le \ell^3/2$, we deduce that the above right hand side is $O(\lambda_{n,1}^*\ell^3+1/n)$, which finishes the proof of the first part. The second conclusion follows immediately from the definition once we use $\|\lambda_{n,-1}^*\|_2=O(\lambda_{n,1}^*\ell^3+1/n)$.
\end{proof}

\begin{lemma}\label{lemma:emp-dual-close-to-population-dual} 
Consider the neighborhood $\Theta_n(\ell,\delta,\eps)$ of the optimizer $\theta_n^*=\theta_n^*(\wh\sigma^2,B)$ of the population dual problem, as defined in \cref{lemma:uniform-concentration}. Denote by $\wh{\theta}_n=\wh{\theta}_n(\wh\sigma^2,\wh{B})$ the optimizer of the empirical dual problem. Assume that $\wh{B}\Pto B$, and that $\ell=o(n^{-1/12})$. For any $\delta,\eta\in(0,1)$, and $\eps=O(n^{-1/4})$, it holds that
$$\P(\wh{\theta}_n \in \Theta_n(\ell,\delta,\eps))\ge  1-\eta,$$
for all sufficiently large $n$.
\end{lemma}

\begin{proof}[Proof of~\cref{lemma:emp-dual-close-to-population-dual}]
It follows from \cref{lemma:uniform-concentration} that for any $n\ge 1$,
\begin{equation*}
\E\sup_{\theta\,\in\,\Theta_n(\ell,\delta,\eps)} |\wh{L}_n(\theta;\wh\sigma^2,\wh{B})-L_n(\theta;\wh\sigma^2,\wh{B})|\lesssim  (\lambda_{n,1}^*)^2\frac{\ell^6}{\sqrt{n}}+\|\lambda_{n,-1}^*\|_2^2\frac{\sqrt{\eps}}{\sqrt{n}}+\lambda_{n,1}^*\|\lambda_{n,-1}^*\|_2\frac{\ell^3}{\sqrt{n}},
\end{equation*}
 where  $\Theta_n(\ell,\delta,\eps):=\{ (\rho,\lambda)\in \Theta(\ell)\text{ such that } |\lambda_1-\lambda_{n,1}^*|\le \delta \lambda_{n,1}^* \text{ and } \|\lambda_{-1}-\lambda^*_{n,-1}\|_2\le \eps \|\lambda_{n,-1}^*\|_2\}$.
Invoking \cref{lemma:one-to-minus-one}, we  reduce the above to the following.
\begin{equation*}
\E\sup_{\theta\,\in\,\Theta_n(\ell,\delta,\eps)} |\wh{L}_n(\theta;\wh\sigma^2,\wh{B})-L_n(\theta;\wh\sigma^2,\wh{B})|\lesssim  (\lambda_{n,1}^*)^2\frac{\ell^6}{\sqrt{n}}+\lambda_{n,1}^*\frac{\ell^3}{n^{3/2}}+\frac{\sqrt{\eps}}{n^{5/2}}.
\end{equation*}
Fix $\eta\in (0,1)$. It follows from above that for some absolute constant $C=C(\eta)>0$, we have
\begin{equation}\label{recall-maximal-ineq}
\P\left(\sup_{\theta\,\in\,\Theta_n(\ell,\delta,\eps)} |\wh{L}_n(\theta;\wh\sigma^2,\wh{B})-L_n(\theta;\wh\sigma^2,\wh{B})|\le C \left((\lambda_{n,1}^*)^2\frac{\ell^6}{\sqrt{n}}+\lambda_{n,1}^*\frac{\ell^3}{n^{3/2}}+\frac{\sqrt{\eps}}{n^{5/2}}\right)\right)\ge 1-\frac{\eta}{2},
\end{equation}
On the other hand, \cref{lemma:one-to-minus-one} also tells us that
$$\sup_{\theta\,\in\,\Theta_n(\ell,\delta,\eps)} |{L}_n(\theta;\wh{\sigma}^2,\wh{B})-L_n(\theta;\wh\sigma^2, B)|\le C\frac{(\lambda_{n,1}^*)^2}{n}\left|\frac{1}{\wh{B}}-\frac{1}{B}\right|.
$$ Consequently,\begin{equation*}
\P\left(\sup_{\theta\,\in\,\Theta_n(\ell,\delta,\eps)} |L_n(\theta;\wh\sigma^2,\wh{B})-L_n(\theta;\wh\sigma^2,B)|\le \delta^2 \frac{(\lambda_{n,1}^*)^2}{16nB^2}\right)\ge 1-\frac{\eta}{2},
\end{equation*}
for all sufficiently large $n$. This, combined with \eqref{recall-maximal-ineq} implies that,
\begin{equation}\label{actual-maximal-ineq}
\P\left(\sup_{\theta\,\in\,\Theta_n(\ell,\delta,\eps)} |\wh{L}_n(\theta;\wh\sigma^2,\wh{B})-L_n(\theta;\wh\sigma^2,B)|\le C\left(  \frac{(\lambda_{n,1}^*)^2\ell^6}{\sqrt{n}}+\frac{\lambda_{n,1}^*\ell^3}{n^{3/2}}+\frac{\sqrt{\eps}}{n^{5/2}}\right)+ \frac{\delta^2(\lambda_{n,1}^*)^2}{16nB^2}\right)\ge 1-\eta,
\end{equation} for all sufficiently large $n$.
Next, we recall from \cref{lemma:strongly-cvx} that 
\begin{equation}\label{recall-strong-convexity}
     L_n(\theta;\wh\sigma^2,B)-L_n(\theta_n^*;\wh\sigma^2,B)
    \ge \frac{\alpha}{4\wh\sigma^2}\left(\|\rho-\rho_n^*\|^2_{L^2(P)}+\|\lambda_{-1}-\lambda_{n,-1}^*\|_2^2\right)+\frac{(\lambda_1-\lambda_{n,1}^*)^2}{4nB^2}-R_n,
\end{equation}
for some $\alpha\in (0,1)$, where $$R_n=\frac{2}{4\wh\sigma^2}\E|(\rho(X)-\rho_n^*(X))\langle \lambda_{-1}-\lambda_{n,-1}^*,v(X,W)\rangle|.$$ 
We define another neighborhood of $\theta_n^*=(\rho_n^*,\lambda_n^*)$, as follows: $$\wt{\Theta}_n:=\{(\rho,\lambda)\in\Theta(\ell):\|\rho-\rho_n^*\|_{L^2(P)}^2+\|\lambda_{-1}-\lambda_{n,-1}^*\|_2^2\le \eps^2\|\lambda_{n,-1}^*\|_2^2\text{ and }|\lambda_1-\lambda_{n,1}^*|\le  \delta\lambda_{n,1}^*\}.$$ For $(\rho,\lambda)\in \wt{\Theta}_n$, it follows from the Cauchy-Schwarz inequality that
\begin{equation}\label{bound-on-remainder}
    R_n\le \frac{1}{4\wh\sigma^2}\left(\|\rho-\rho_n^*\|_{L^2(P)}^2+\|\langle\lambda_{n,-1}-\lambda_{n,-1}^*,v(X,W)\rangle\|_2^2\right) \le \frac{\beta}{4\wh\sigma^2}\eps^2\|\lambda_{n,-1}^*\|_2^2,
\end{equation}where $\beta=\max\{1,\lambda_{\max}\}$, where $\lambda_{\max}$ is the largest eigenvalue of $A=\E[v(X,W)v(X,W)^\top]$.
It is immediate from definition that $\wt{\Theta}_n\subseteq \Theta_n(\ell,\delta,\eps)$, and that for every $\theta$ on the boundary $\del \wt{\Theta}_n$,  $$\|\rho-\rho_n^*\|_{L^2(P)}^2+\|\lambda_{-1}-\lambda_{n,-1}^*\|_2^2= \eps^2\|\lambda_{n,-1}^*\|_2^2,\quad \text{and}\quad |\lambda_1-\lambda_{n,1}^*|= \delta\lambda_{n,1}^*.$$ We can now combine \eqref{actual-maximal-ineq}, \eqref{recall-strong-convexity} and \eqref{bound-on-remainder} to conclude that, for any $\theta\in\del\wt{\Theta}_n$, 
\begin{align*}
      &\frac{\alpha}{4\wh\sigma^2}\eps^2\|\lambda_{n,-1}^*\|_2^2+\frac{\delta^2(\lambda_{n,1}^*)^2}{4nB^2}-\frac{\beta}{4\wh\sigma^2}\eps^2\|\lambda_{n,-1}^*\|_2^2\\
      &\le L_n(\theta;\wh\sigma^2,B)-L_n(\theta_n^*;\wh\sigma^2,B)\tag{using \eqref{recall-strong-convexity} and \eqref{bound-on-remainder}}\\
      &\le \wh{L}_n(\theta;\wh\sigma^2,\wh{B})-\wh{L}_n(\theta_n^*;\wh\sigma^2,\wh{B}) + 2\sup_{\theta\in \Theta_n(\ell,\delta,\eps)}|\wh{L}_n(\theta;\wh\sigma^2,\wh{B})-L_n(\theta;\wh\sigma^2,B)|\label{master-calculation}\numberthis\\
      &\le \wh{L}_n(\theta;\wh\sigma^2,\wh{B})-\wh{L}_n(\theta_n^*;\wh\sigma^2,\wh{B}) +C \left((\lambda_{n,1}^*)^2\frac{\ell^6}{\sqrt{n}}+\lambda_{n,1}^*\frac{\ell^3}{n^{3/2}}+\frac{\sqrt{\eps}}{n^{5/2}}\right)+\delta^2 \frac{(\lambda_{n,1}^*)^2}{8nB^2},\tag{using~\eqref{actual-maximal-ineq}}
\end{align*}
with probability at least $1-\eta$, for any $n\ge 1$. We can rewrite the above inequality as $\wh{L}_n(\theta_n^*;\wh\sigma^2,\wh{B})\le \wh{L}_n(\theta;\wh\sigma^2,\wh{B})+\Delta_n$ and we show below that $\Delta_n<0$ for all sufficiently large $n$. This implies, by taking infimum over $\theta\in\del\wt{\Theta}_n$, that $$\P\left(\wh{L}_n(\theta_n^*;\wh\sigma^2,\wh{B})\le \inf_{\theta\in \partial \wt{\Theta}_n}\wh{L}_n(\theta;\wh\sigma^2,\wh{B})\right)\ge 1-\eta,$$
for all sufficiently large $n$. Finally, note that since the empirical loss function $\wh{L}_n(\theta;\wh\sigma^2,\wh{B})$ is convex in $\theta$, the above event implies that $\wh{\theta}_n=\argmin_\theta \wh{L}_n(\theta;\wh\sigma^2,\wh{B})$ must be inside $\wt{\Theta}_n$. Thus, $$\P\left(\widehat{\theta}_n\in \wt{\Theta}_n\right)\ge 1-\eta,$$
for all sufficiently large $n$, as desired to show.  

To complete the proof, we need to show that $\Delta_n<0$ for all sufficiently large $n$.
Observe that, 
\begin{align*}
   \Delta_n&=C \left((\lambda_{n,1}^*)^2\frac{\ell^6}{\sqrt{n}}+\lambda_{n,1}^*\frac{\ell^3}{n^{3/2}}+\frac{\sqrt{\eps}}{n^{5/2}}\right)+\delta^2 \frac{(\lambda_{n,1}^*)^2}{8nB^2}-\left(\frac{\delta^2(\lambda_{n,1}^*)^2}{4nB^2}-\frac{\beta-\alpha}{4\wh\sigma^2}\eps^2\|\lambda_{n,-1}^*\|_2^2\right) \\
   &\le C \left((\lambda_{n,1}^*)^2\frac{\ell^6}{\sqrt{n}}+\lambda_{n,1}^*\frac{\ell^3}{n^{3/2}}+\frac{\sqrt{\eps}}{n^{5/2}}\right)-\delta^2 \frac{(\lambda_{n,1}^*)^2}{8nB^2}+c_0\eps^2 \left(\lambda_{n,1}^*\ell^3+\frac{1}{n}\right)^2 \\
   &= \frac{(\lambda_{n,1}^*)^2}{n}\left(C\ell^6\sqrt{n} + c_0 \eps^2 \ell^6 n -\frac{\delta^2}{8B^2}+\frac{C\ell^3}{\lambda_{n,1}^*\sqrt{n}} + 2c_0 \eps^2 \frac{\ell^3}{\lambda_{n,1}^*} +\frac{C\sqrt{\eps}}{n^{3/2}(\lambda_{n,1}^*)^2}+\frac{c_0\eps^2}{n(\lambda_{n,1}^*)^2}\right).
\end{align*}
We show in \cref{lemma:minimax-bias-rate} that $\lambda_{n,1}^*=O(n^{-3/7})$. In light of this, we can write
$$\limsup_{n\to\infty} \frac{\Delta_n}{(\lambda_{n,1}^*)^2/n}\le \lim_{n\to\infty} \left(C\ell^6\sqrt{n} + c_0 \eps^2 \ell^6 n -\frac{\delta^2}{8B^2}\right)\le - \frac{\delta^2}{8B^2}<0,$$
since $\ell^6\sqrt{n}=o(1)$ and $\eps^2 \ell^6n = \eps^2 \sqrt{n} \cdot \ell^6\sqrt{n} =o(1)$ as well. This completes the proof.\end{proof}

\begin{lemma}\label{lemma:minimax-bias-rate}
   Under the conditions of \cref{thm:validity}, $$\wh\lambda_{n,1}(\wh\sigma^2, \wh{B})=\O(n^{-3/7}),\quad  \lambda_{n,1}^*(\wh\sigma^2, {B})=\O(n^{-3/7}),\quad\text{and}\quad \sum_{i=1}^n\wh\g_i^2=\O(n^{-6/7})$$
\end{lemma}

\begin{proof}[Proof of \cref{lemma:minimax-bias-rate}]
    Recall that $\wh\lambda_{n,1}(\wh\sigma^2, \wh{B}) = 2\wh{B}^2 \wh{t}_n$, where $\wh{B}^2\wh{t}_n^{\,2}$ is bounded above by the worst-case MSE of any linear estimator of the form $\sum_{i=1}^n \g_i Y_i$. In particular, there exists a kernel regression estimator that has worst-case MSE (supremum taken over our function class $\mathcal{M}_{\wh{B}}$) of the order of $\O(n^{-6/7})$. To make it precise, consider the kernel regression estimator
    $$\left(\wh\tau(h_n), \wh\alpha, \wh\beta \right)=\argmin_{\tau, \alpha, \beta_1,\beta_0 \in \mathbb{R}}\ \sum_{i=1}^n K\left(\frac{X_i }{h_n}\right)\left(Y_i-\alpha-\tau W_i-\beta_0 X_i^{-}-\beta_1 X_i^{+}-\beta_2 X_i^2 \right)^2.$$ For simplicity, we use the triangular kernel $K(x)=1\{|x|\le 1\}$.
    It follows along the lines of the proof of~\citep[Proposition 8.1]{wager2024causal} that $\wh\tau(h_n)$ has worst-case squared curvature bias bounded above by $\wh{B}h_n^6$, and sampling noise of the order of $1/nh_n$. With $h_n=O(n^{-1/7})$, this estimator has MSE of the order of $n^{-6/7}$. Consequently, $\wh\lambda_{n,1}=\O(n^{-3/7})$ and $\|\wh\g_n\|_2^2=\O(n^{-6/7})$.
    The proof for $\lambda_{n,1}^*$ is analogous. 
\end{proof}

\begin{lemma}\label{lemma:bias-ok}
Under the conditions of \cref{thm:main-CLT}, it holds that $\wh{b}/{s}(\wh\g)=\O(1)$ as $n\to\infty$, where ${s}^2(\wh\g)=\sum_{i=1}^n\wh\g_i^2\sigma_i^2$, with $\sigma_i^2=\var(Y_i\mid X_i)$.
\end{lemma}

\begin{proof}[Proof of \cref{lemma:bias-ok}] Assume $c=0$ without loss of generality. First, we can replace ${s}^2(\wh\g)=\sum_{i=1}^n \wh\g_i^2\sigma_i^2$ with $\wh{V}_n:=\wh\sigma^2\sum_{i=1}^n \wh\g_i^2$, using the fact that $0<\sigma_{\min}\le \sigma_i\le \sigma_{\max}<\infty$. Recall that, ignoring cross-fitting notation as mentioned in the beginning of this section, we define
$$(\wh{\g}_n,\wh{t}_n):=\argmin_{\g, \, t} \left(\wh{B}^2 t^2 +\wh{\sigma}^2\sum_{i=1}^n \wh\g_i^2\right)\quad\text{subject to}\quad \sup_{\rho\in \rhos}\sum_{i=1}^n \g_i\rho(X_i)\le t,$$ where $\wh{B}$ and $\wh{\sigma}^2$ are fit externally, independent of the data $\{X_i,Y_i\}_{i=1}^n$. The worst-case bias is given by
$\wh{b}=\wh{B}\wh{t}_n$, and thus the goal here is to show that $\wh{t}_n^{\,2}/\wh{V}_n  = \O(1)$. Recall the population-level problem with externally fit nuisance parameters: $$(\g^*_n, t_n^*):=\argmin_{\g,t} \left\{\wh{B}^2 t^2 + \frac{1}{n}\wh{\sigma}^2\E[\g^2(X)]\right\} \quad\text{subject to}\quad\sup_{\rho\in\rhos}\E [\g(X)\rho(X)]\le t.$$
It suffices to prove that $t_n^*/\sqrt{V_n^*}=\O(1)$, where $V_n^* = \wh{\sigma}^2 \E [\g_n^{*,2}(X)]/n$, since we have consistency (i.e., that $\wh{t}_n/t^*_n\Pto 1$ and $\wh{V}_n/V_n^*\Pto 1$).
It follows from \cref{lemma:minimax-bias-rate} that  
$\wh{B}^2t_n^{\,*,\,2}+V_n^*\le_p a_+ n^{-6/7}$ for some $a_+\in (0,\infty)$, for all sufficiently large $n$. On the other hand, it follows from the minimax lower bound (see, {\it e.g.},~\citep{Cheng-Fan-Marron}) that $\wh{B}^2 t_n^{\,*,\,2}+V_n^*\ge_p a_{-} n^{-6/7}$ for some $0<a_{-}\le a_{+}$ for all sufficiently large $n$. Therefore, to claim that $t_n^*/\sqrt{V_n^*}=\O(1)$, it suffices to show that there exists some $c>0$ such that $V_n^*\ge_p c n^{-6/7}$ for all sufficiently large $n$.

To prove this by contradiction, assume that for any $c>0$, we have $V_n^*\le_p c n^{-6/7}$ infinitely often. We  examine the MSE of $\g_n^*$ for the objective $$R_m(\g):=B^2\left(\sup_{\rho\in\rhos}\E \g(X)\rho(X)\right)^2+\frac{1}{m}\E[\g^2(X)],$$ where $m=n/k$ (assume without loss of generality that $n$ is divisible by $k$). Note that\begin{align*}
R_m(\g_n^*)&=\wh{B}^2\left(\sup_{\rho\in\rhos}\E \g_n^*(X)\rho(X)\right)^2+\frac{1}{m}\E[\g^{*,2}_n(X)]\\
&= \wh{B}^2(t_n^*)^2+k V_n^*\le_p (a_+ + k c)k^{-6/7} m^{-6/7}.
\end{align*}
Choose $k$ large enough so that $a_+k^{-6/7}\le a_{-}/2$, and then choose $c>0$ small enough such that $ck^{1/7}<a_{-}/2$. This, along with the above display, implies that $R_m(\g_m^*)\le R_m(\g_n^*)\le_p a_{-} m^{-6/7}$, which contradicts the minimax lower bound (we choose $n$ is large enough so that $m=n/k$ satisfies the
minimax lower bound $\wh{B}^2 t_m^{\,*,\,2}+V_m^*\ge_p a_{-} m^{-6/7}$). This contradiction completes the proof.
\end{proof}

\begin{lemma}\label{lemma:Gaussian-CDF}
    Suppose that $Z\sim N(0,1)$. Then, for any $s$, $h$, $t$ and $\delta>0$, we have $$\inf_{|b|\,\le\, t+\delta s}\P(|b+s Z|\le h) \ge \inf_{|b|\,\le\, t}\P(|b+s Z|\le h) -\delta.$$
\end{lemma}

\begin{proof}[Proof of \cref{lemma:Gaussian-CDF}]
    Fix any $s$, $h$, $t$ and $\delta>0$. Denote $$g(b)=\P(|b+s Z|\le h)=\Phi\left(\frac{h-b}{s}\right)-\Phi\left(\frac{-h-b}{s}\right).$$
    Note that, $$g'(b)=-\frac{1}{s}\left(\phi\left(\frac{h-b}{s}\right)-\phi\left(\frac{-h-b}{s}\right)\right)\implies \|g'\|_\infty \le \frac{2\phi(0)}{s}\le \frac{1}{s}.$$
    Next, fix any $b_1$ with $|b_1|\le t+\delta s$. There exists $b_0$ such that $|b_0|\le t$ and $|b_1-b_0|\le \delta s$. Consequently, $|g(b_1)-g(b_0)|\le \|g'\|_\infty |b_1-b_0|\le \delta$. This implies that $g(b_1)\ge g(b_0)-\delta\ge \inf_{|b|\le t} g(b) -\delta$. Taking infinimum  over $b_1$, we get the desired conclusion.
\end{proof}

\begin{lemma}\label{lemma:ehw}
    Under the conditions of \cref{thm:main-CLT}, it holds that $\wh{s}(\wh\g)/{s}(\wh\g)\Pto 1$ as $n\to\infty$, where ${s}^2(\wh\g)=\sum_{i=1}^n\wh\g_i^2\sigma_i^2$ with $\sigma_i^2=\var(Y_i\mid X_i)$, and $\wh{s}^2(\wh\g)=\sum_{i=1}^n \wh\g_i^2\wh\sigma_i^2$ with $\wh\sigma_i^2=(Y_i - \wh\mu_{W_i}(X_i))^2$ as defined in \cref{algo:PLRD}.
\end{lemma}

\begin{proof}[Proof of \cref{lemma:ehw}]
   In \cref{algo:PLRD}, we use sample-splitting and fit regression on either side of the threshold; here we assume for simplicity in notation that we fit the regression of $Y_i$ on $X_i$ globally (our argument applied on either side of the threshold will then complete the proof). Also assume without loss of generality that $c=0$. Write $\wt{X}_i=(1,X_i)$, $\wh\eps_i=Y_i-\wh\mu(X_i)=Y_i-\wt{X}_i\wh\beta$ and $\eps_i^* = Y_i - \mu(X_i)=Y_i-\wt{X}_i\beta^*$ where $\wt{X}\beta^*$ is the best linear projection of $\mu(X)=\E[Y\mid X]$ onto $\text{span}\{1,X\}$. It follows from \citep{white1980heteroskedasticity} that $\wh\beta\Pto \beta^*$.
   Next, write
   \begin{equation}\label{eqn:sum-of-sq}
       \sum_{i=1}^n\wh\g_i^2\wh\eps_i^2 = \sum_{i=1}^n\wh\g_i^2\eps_i^{*2} + 2\sum_{i=1}^n \wh\g_i^2 \eps_i^{*}\wt{X}_i^\top(\beta^*-\wh\beta)+\sum_{i=1}^n \wh\g_i^2(\wt{X}_i^\top(\beta^*-\wh\beta))^2.
   \end{equation}
   On the other hand, note that $\sum_{i=1}^n \wh\g_i^2\sigma_i^2\ge \sigma^2_{\min}\sum_{i=1}^n \wh\g_i^2$. Therefore,
   $$\frac{\sum_{i=1}^n \wh\g_i^2(\wt{X}_i^\top(\beta^*-\wh\beta))^2}{\sum_{i=1}^n \wh\g_i^2\sigma_i^2}\lesssim \left\|(\wh\beta-\beta^*)^\top\frac{ |\wh\g_i| \wt{X}_i}{\|\wh\g_i\|_2}\right\|_2^2=\o(1),$$
   since $\wh\g_i=0$ for $|X_i|>\ell$. For the cross-term, we apply Cauchy-Schwarz inequality to write
   $$\frac{\left|\sum_{i=1}^n \wh\g_i^2 \eps_i^{*}\wt{X}_i^\top(\beta^*-\wh\beta)\right|}{\sum_{i=1}^n \wh\g_i^2\sigma_i^2}\le \left(\frac{\sum_{i=1}^n\wh\g_i^2\eps_i^{*2}}{\sum_{i=1}^n\wh\g_i^2\sigma_i^2}\right)^{1/2}\left\|(\wh\beta-\beta^*)^\top\frac{ |\wh\g_i| \wt{X}_i}{\|\wh\g_i\|_2}\right\|_2.$$
   We can handle the second term exactly as above. For the first part, we show below that
   \begin{equation}\label{eqn:weighted-LLN}
       \sum_{i=1}^n \wh\g_i^2 \eps_i^{*2}\;\big/\;\sum_{i=1}^n \wh\g_i^2\sigma_i^2\Pto 1.
   \end{equation}
   As a consequence, we can continue from \eqref{eqn:sum-of-sq} to write
   $$\frac{\sum_{i=1}^n\wh\g_i^2\wh\eps_i^2}{\sum_{i=1}^n\wh\g_i^2\sigma_i^2}=\frac{\sum_{i=1}^n\wh\g_i^2\eps_i^{*2}}{\sum_{i=1}^n\wh\g_i^2\sigma_i^2}+\o(1)+\o(1),$$
   which completes the proof.
   To prove \eqref{eqn:weighted-LLN}, fix any $\eps>0$ and $\eta>0$. Define $a_{n,i}=\wh\g_i^2/\sum_{i=1}^n \wh\g_i^2$, $1\le i\le n$, and write $\Delta_i=\eps_i^{*2}-\sigma_i^2$. It follows from the assumption $\sup_x \E[|Y_i-\mu(X_i)|^q\mid X_i=x]<\infty$ and Jensen's inequality that $\E[|\Delta_i|^{q/2} | X_i=x]\le C$ for some universal constant $C\in (0,\infty)$. Then,
   \begin{align*}
      &\E\left[\left| \sum_{i=1}^n a_{n,i}\Delta_i\ind{|\Delta_i|>M}\right|\mid X_1,\dots,X_n\right]\le \sum_{i=1}^n a_{n,i}\E\left[\left| \Delta_i\right|\ind{|\Delta_i|>M}\mid X_1,\dots,X_n\right]\\
     \tag{Hölder's inequality} &\le \sum_{i=1}^n a_{n,i}\left(\E\left[ \left|\Delta_i\right|^{q/2}\mid X_i\right]\right)^{2/q}\P\left(|\Delta_i|>M \mid X_i\right)^{1-2/q}
      \\
   \tag{Markov's inequality}   &\le \sum_{i=1}^n a_{n,i}\frac{1}{M^{q/2-1}}\E\left[ \left|\Delta_i\right|^{q/2}\mid X_i\right]\\
 &    \le \frac{C}{M^{q/2-1}}.
   \end{align*}
   Since $q>2$, we can pick $M$ large enough so that  $$\P\left(\left| \sum_{i=1}^n a_{n,i}\Delta_i\ind{|\Delta_i|>M}\right|>\eps/2\right)<\eta/2,$$
   and $|\E[\Delta_i\ind{|\Delta_i|\le M}|\le\eps/4$ (the latter is possible since $\lim_{M\to\infty}\E[\Delta_i\ind{|\Delta_i|\le M}]= \E[\Delta_i]=0$ by Dominated Convergence Theorem).
   Consequently, by Chebyshev's inequality, 
   $$\P\left(\left| \sum_{i=1}^n a_{n,i}\Delta_i\ind{|\Delta_i|\le M}\right|>\eps/2\,\bigg|\, \{X_i\}\right)\lesssim \sum_{i=1}^n a_{n,i}^2 M^2 \le M^2\max_{1\le i\le n}a_{n,i}<\eta/2,$$
   for all sufficiently large $n$, since we have $\max_{i\le n}a_{n,i}=\o(1)$ from \cref{coro:lindeberg}. Putting these together, we obtain $\sum_{i=1}^n \wh\g_i^2(\eps_i^{*2}-\sigma_i^2)/\sum_{i=1}^n \wh\g_i^2=\o(1)$. Since $0\le \sigma^2_{\min}\le \sigma_i^2$, this finishes the proof of  \eqref{eqn:weighted-LLN}.
\end{proof}

\end{document}